\documentclass[11pt,a4paper]{article}

\usepackage{amsmath,amsfonts,amssymb,amsthm}
\usepackage{amsthm}
\usepackage{graphicx,color}
\usepackage{boxedminipage}
\usepackage[linesnumbered, boxed, algosection,noline]{algorithm2e}
\usepackage{framed}
\usepackage{thmtools}
\usepackage{thm-restate}
\usepackage{xspace}
\usepackage{vmargin}
\setmarginsrb{1in}{1in}{1in}{1in}{0pt}{0pt}{0pt}{6mm}
\usepackage{todonotes}
 \usepackage[pdftex, plainpages = false, pdfpagelabels, 
                 bookmarks=false,
                 bookmarksopen = true,
                 bookmarksnumbered = true,
                 breaklinks = true,
                 linktocpage,
                 pagebackref,
                 colorlinks = true,  
                 linkcolor = blue,
                 urlcolor  = blue,
                 citecolor = red,
                 anchorcolor = green,
                 hyperindex = true,
                 hyperfigures
                 ]{hyperref} 
 \usepackage{xifthen}
 \usepackage{tabularx}
\usepackage{tikz} 
 \usetikzlibrary{calc}

\newcommand{\NPC}{\textsf{NP}-complete}

\DeclareMathOperator{\operatorClassNP}{{\sf NP}}
\newcommand{\classNP}{\ensuremath{\operatorClassNP}}

\DeclareMathOperator{\operatorClassFPT}{{\sf FPT}\xspace}
\newcommand{\classFPT}{\ensuremath{\operatorClassFPT}\xspace}
\DeclareMathOperator{\operatorClassW}{{\sf W}}
\newcommand{\classW}[1]{\ensuremath{\operatorClassW[#1]}}

\newcommand{\Oh}{\mathcal{O}}

\newcommand{\bran}[1]{branchable\xspace}

\newtheorem{theorem}{Theorem}
\newtheorem{lemma}{Lemma}

\newtheorem{corollary}{Corollary}

\newtheorem{observation}{Observation}
\newtheorem{proposition}{Proposition}

\theoremstyle{definition}

\newcommand{\cost}{\omega}
\newcommand{\weight}{W}

\newcommand{\pname}{\textsc}
\newcommand{\ProblemFormat}[1]{\pname{#1}}
\newcommand{\ProblemIndex}[1]{\index{problem!\ProblemFormat{#1}}}
\newcommand{\ProblemName}[1]{\ProblemFormat{#1}\ProblemIndex{#1}{}\xspace}

\newcommand{\probSCAtwo}{\ProblemName{Structured $2$-Connectivity Augmentation}}
\newcommand{\probWSCA}{\ProblemName{Structured  $k$-Connectivity Augmentation}}

\newcommand{\probStrucAugm}{\ProblemName{Structured  Connectivity Augmentation}}
\newcommand{\probStrucAugmtwo}{\ProblemName{Structured  $2$-Connectivity Augmentation}}



\newlength{\RoundedBoxWidth}
\newsavebox{\GrayRoundedBox}
\newenvironment{GrayBox}[1]%
   {\setlength{\RoundedBoxWidth}{.93\textwidth}
    \def\boxheading{#1}
    \begin{lrbox}{\GrayRoundedBox}
       \begin{minipage}{\RoundedBoxWidth}}%
   {   \end{minipage}
    \end{lrbox}
    \begin{center}
    \begin{tikzpicture}%
       \node(Text)[draw=black!20,fill=white,rounded corners,%
             inner sep=2ex,text width=\RoundedBoxWidth]%
             {\usebox{\GrayRoundedBox}};
        \coordinate(x) at (current bounding box.north west);
        \node [draw=white,rectangle,inner sep=3pt,anchor=north west,fill=white] 
        at ($(x)+(6pt,.75em)$) {\boxheading};
    \end{tikzpicture}
    \end{center}}     

\newenvironment{defproblemx}[2][]{\noindent\ignorespaces%
                                \FrameSep=6pt%
                                \parindent=0pt%
                \vspace*{-1.5em}
                \ifthenelse{\isempty{#1}}{%
                  \begin{GrayBox}{\textsc{#2}}%
                }{%
                  \begin{GrayBox}{\textsc{#2} parameterized by~{#1}}%
                }
                \begin{tabular*}{\textwidth}{@{\hspace{.1em}} >{\itshape} p{1.8cm} p{0.8\textwidth} @{}}%
            }{
                \end{tabular*}%
                \end{GrayBox}%
                \ignorespacesafterend
            }

\newcommand{\defproblema}[3]{
  \begin{defproblemx}{#1}
    Input:  & #2 \\
    Task: & #3
  \end{defproblemx}
}%

\begin{document}

\title{Structured Connectivity Augmentation\thanks{The two first authors have been supported by the Research Council of Norway via the projects ``CLASSIS'' and ``MULTIVAL". The third author has been supported by project ``DEMOGRAPH" (ANR-16-CE40-0028). \newline Emails of authors: \texttt{\{fedor.fomin, petr.golovach\}@ii.uib.no}, \texttt{sedthilk@thilikos.info} .}
}

\author{
Fedor V. Fomin\thanks{
Department of Informatics, University of Bergen, Norway.} \addtocounter{footnote}{-1}
\and
Petr A. Golovach\footnotemark{}
\and 
Dimitrios M. Thilikos\thanks{AlGCo project, CNRS, LIRMM, France.}~\thanks{Department of Mathematics National and Kapodistrian University of Athens, Greece.}}

\date{}

\maketitle

\begin{abstract}
\noindent We initiate the algorithmic study of the following ``structured augmentation'' question: is it possible to increase the connectivity of a given graph $G$ by superposing it with another given graph $H$? More precisely, 
graph $F$ is the  superposition of $G$ and $H$ with respect to  injective mapping 
$\varphi\colon V(H)\rightarrow V(G)$  
if every edge $uv$ of $F$ is either an edge of $G$, or    $\varphi^{-1}(u)\varphi^{-1}(v)$ is an edge of $H$.
Thus $F$ contains both $G$   and $H$ as subgraphs, and the edge set of $F$ is the union of the  edge sets of $G$ and $\varphi(H)$. 
We consider the following optimization problem. Given graphs $G$, $H$,  and  a weight function $\omega$ assigning   non-negative weights to   pairs of vertices of $V(G)$, the task is  to find
$\varphi$  of minimum weight  $ \omega(\varphi)=\sum_{xy\in E(H)}\cost(\varphi(x)\varphi(y))$ such that the edge connectivity of the
superposition $F$ of $G$ and $H$ with respect to  
$\varphi$ is higher than the edge connectivity of $G$. 
 Our main result is the following ``dichotomy'' complexity classification.  
  We say that a  class of graphs $\mathcal{C}$ has \emph{bounded vertex-cover number}, if there is a constant $t$ depending on  $\mathcal{C}$ only such that the vertex-cover number 
    of every graph from  $\mathcal{C}$ does not exceed $t$.  We show that for every class of graphs  $\mathcal{C}$  with bounded vertex-cover number, the problems of superposing into a connected graph $F$ and to 2-edge connected graph $F$,  are solvable in polynomial time    when  $H\in\mathcal{C}$. On the other hand, for \emph{any} hereditary class $\mathcal{C}$
  with unbounded vertex-cover number, both problems are  \classNP-hard when  $H\in\mathcal{C}$. 
For  the unweighted variants of    structured augmentation problems, i.e. the problems where the task is to identify whether there is a superposition of graphs of required connectivity, we provide 
necessary and sufficient combinatorial conditions on the existence of such superpositions. 
 These conditions imply  polynomial time algorithms solving the unweighted variants of the problems.

\end{abstract}



\section{Introduction}\label{sec:intro}

In connectivity augmentation problems, the input is a (multi) graph and the objective is to increase edge or vertex connectivity by adding the minimum number (weight) of additional edges, called links. 
This is a fundamental combinatorial problem  with a number of important applications, we refer to the books of Nagamochi  and Ibaraki 
\cite{Nagamochi08}  and Frank \cite{Frank11} for a detailed introduction to the topic. In this paper we initiate the study of a ``structural'' connectivity augmentation problem, where the set of additional edges should satisfy some additional constrains. For example, such constrains can be that all new edges should be visible from one vertex, i.e. the new set of edges forms a star, forms a cycle, or can be controlled from a small set of vertices, i.e. the graph formed by the additional edges has a small vertex cover.

 It is convenient to model such an augmentation problem as a \emph{graph superposition} problem. 
Let $G$ and $H$ be simple graphs (i.e. graphs  without loops and multiple edges), $|V(G)|\geq |V(H)|$,  and let $\varphi\colon V(H)\rightarrow V(G)$ be an injective mapping of the vertices of $H$ to the set of vertices of $V(G)$. We say that a simple  graph $F$ is the \emph{superposition of $G$ and $H$ with respect to $\varphi$} and write $F=G\oplus_\varphi H$ if $V(F)=V(G)$ and two distinct vertices $u,v\in V(F)$ are adjacent in $F$ if and only if $uv\in E(G)$ or $u,v\in \varphi(V(H))$ and $\varphi^{-1}(u)\varphi^{-1}(v)\in E(H)$. See Fig.~\ref{fig:superposition} for an example. Thus graph $F$ contains $G$ and $H$ as subgraphs, and the edge set of $F$ is the union of the  edge sets of $G$ and $\varphi(H)$. 

\begin{figure}[t]
\begin{center}
\includegraphics[scale=.35]{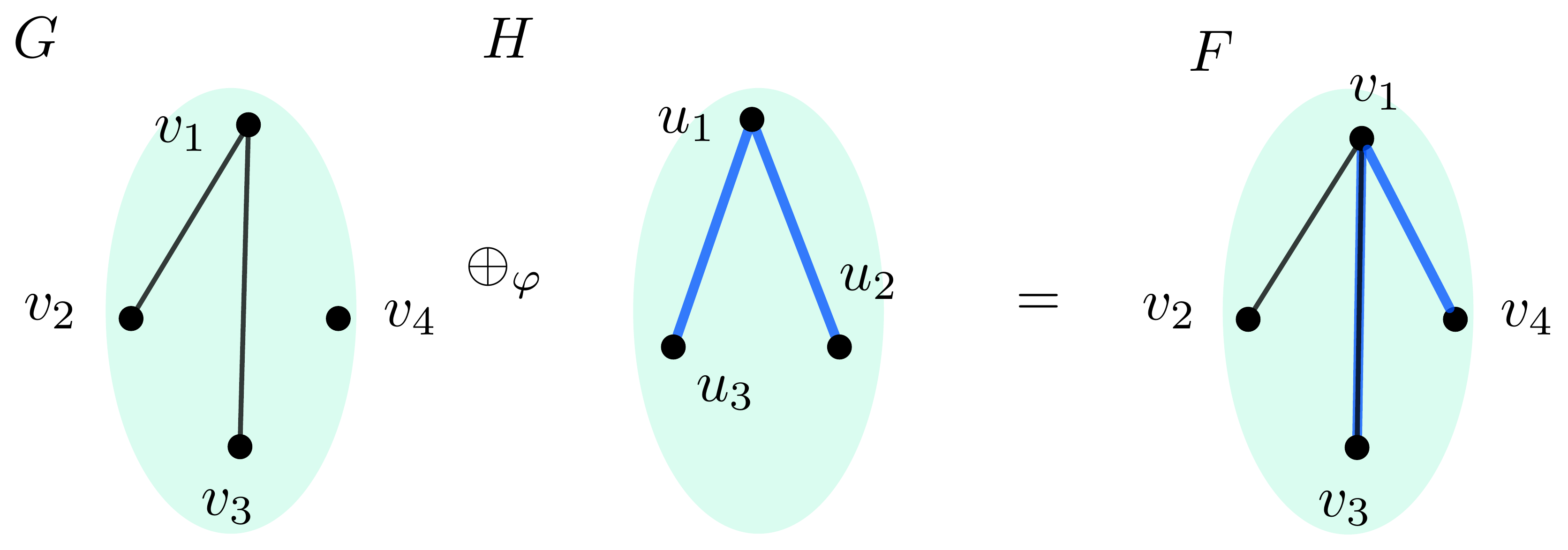}

\end{center}
\caption{For injective mapping  $\varphi\colon V(H)\rightarrow V(G)$ such that $\varphi (u_1)=v_1$,  $\varphi (u_2)=v_4$, and $\varphi (u_3)=v_3$, we have  $F=G\oplus_\varphi H$.}
\label{fig:superposition}
\end{figure} 

We study the algorithmic problem of increasing the edge-connectivity of graph $G$ by superposing it with a graph $H$. We are interested in the weighted variant of the problem, where for every pair of   vertices $v$ and $u$ of $G$, mapping the endpoints of an edge of $H$ to  $u $ and $v$ has a specified weight $\cost(uv)$. We consider the following problem.

\defproblema{\probStrucAugm}%
{Graphs $G$ and $H$, a weight function $\cost\colon \binom{V(G)}{2}\rightarrow \mathbb{N}_0 $,   and a nonnegative integer $\weight$.}%
{Decide whether there is an injective  mapping $\varphi\colon V(H)\rightarrow V(G)$ such that graph $F=G\oplus_{\varphi}H$ is  connected and the \emph{weight} of the mapping $\cost(\varphi)=\sum_{xy\in E(H)}\cost(\varphi(x)\varphi(y))\leq \weight$.}

We also study the problem of obtaining a $2$-edge connected graph $F$ by superposing graphs $G$ and $H$. More precisely, we consider the following problem.

\defproblema{\probStrucAugmtwo}%
{Connected graph $G$ and a graph $H$, a weight function $\cost\colon \binom{V(G)}{2}\rightarrow \mathbb{N}_0 $  and a nonnegative integer $\weight$.}%
{Decide whether there is an injective mapping  $\varphi\colon V(H)\rightarrow V(G)$ of weight at most $W$ such that $F=G\oplus_{\varphi}H$ is  $2$-edge connected.}

 \medskip\noindent\textbf{Our results.} Our main result is the following ``dichotomy'' complexity classification of structured augmentation problems. 
  We say that a  class of graphs $\mathcal{C}$ has \emph{bounded vertex-cover number}, if there is a constant $t$ depending on  $\mathcal{C}$ only such that the vertex-cover number 
    of every graph from  $\mathcal{C}$ does not exceed $t$.  We show that for every class of graphs  $\mathcal{C}$  with bounded vertex-cover number, 
  \probStrucAugm and \probStrucAugmtwo are solvable in polynomial time    when  $H\in\mathcal{C}$. We complement this result by showing that for \emph{any} hereditary class $\mathcal{C}$
  with unbounded vertex-cover number, both problems are  \classNP-complete when  $H\in\mathcal{C}$. 
  Thus for any hereditary class $\mathcal{C}$ both problems with  $H\in\mathcal{C}$ are  \classNP-complete if and only if  $\mathcal{C}$ has  {unbounded vertex-cover number}.

  The running times of our  algorithms solving  \probStrucAugm and \probStrucAugmtwo are of the form   $|V(G)|^{\Oh(f(t))}\cdot \log W$, where $f$ is some function and $t$ is the vertex cover of $H$. Thus our algorithms are not fixed-parameter tractable when $t$ is the parameter. We show that from the perspective of parameterized complexity,  this situation is unavoidable. More precisely, we 
  show that both problems are \classW{1}-hard when parameterized by  $t$.  We refer to the book of Downey and Fellows~\cite{DowneyF13}
for an introduction to parameterized complexity. 

 We also consider the unweighted variants of   \probStrucAugm and \probStrucAugmtwo.  In these cases, the weight  is $\cost(uv)=0$ for every pair of vertices of $G$ and $W=0$. The task is to identify whether there is a superposition of graphs $G$ and $H$ of edge connectivity $1$ or $2$, correspondingly.   
Here we obtain   necessary and sufficient combinatorial conditions of the existence of an injective function $\varphi$ such that $F=G\oplus_\varphi H$ is edge $k$-connected provided that $G$ is edge $(k-1)$-connected, $k=1, 2$. These conditions imply  polynomial time algorithms solving the unweighted variants of the problems. 
    
\medskip
\noindent\textbf{Related work.} The problem of increasing graph connectivity by adding additional edges is the classic and well-studied problem. 
It was first studied by Eswaran and Tarjan~\cite{EswaranT76} and Plesnik \cite{Plesnik76} who showed that increasing the edge connectivity of a given graph to 2 by adding minimum number of additional augmenting edges is polynomial time solvable. Subsequent work in \cite{Watanabe198796,Frank92} showed that this problem is also polynomial time solvable for any given target value of edge connectivity to be achieved. However, if the set of augmenting edges is restricted, that is, there are pairs of vertices in the graph which do not constitute a new edge, or if the augmenting edges have (non-identical) weights on them, then the problem of computing the minimum size (or weight) augmenting set is {\NPC}~\cite{EswaranT76}. 
 Augmentation problems with constraints like simplicity-preserving augmentations, augmentations with partition constraints, or planarity requirements can be found in the literature, see the book of 
 Nagamochi  and Ibaraki 
\cite{Nagamochi08} for further references. 

Strongly relevant to structural augmentation is the \textsc{Minimum Star Augmentation} problem, see e.g. 
\cite[Section~3.3.3]{Nagamochi08} and  \cite{TiborZ03}. Here one wants to increase the edge-connectivity of a given graph by adding a new vertices and connecting it with a small number of edges to the remaining vertices of the graph. In our setting this corresponds to the case of graph $G$ having an isolate vertex, and graph $H$ being a star (a tree with vertex-cover number $1$).  Tibor and Szigeti  \cite{TiborZ03} studied a generalization of this problem where one wants to make a graph edge $r$-connected by attaching $p$ stars of specified degrees. In particular, they provided  combinatorial conditions which are necessary and sufficient for such an augmentation. Again, this problem can be seen as a special case of structural augmentation,  where graph $G$ has $p$ isolated vertices and graph $H$ is the union of stars of specified degrees.

\section{Preliminaries}\label{sec:defs}
 We consider only finite undirected graphs. For a graph $G$, $\binom{V(G)}{2}$ denotes the set of unordered pairs of distinct vertices of $G$. For uniformity, we denote the elements of  $\binom{V(G)}{2}$ in the same way as edges, i.e., we write $uv\in \binom{V(G)}{2}$.
A subgraph $H$ of $G$ is \emph{spanning} if $V(H)=V(G)$.
For a graph $G$ and a subset $U\subseteq V(G)$ of vertices, we write $G[U]$ to denote the subgraph of $G$ induced by $U$. 
We write $G-U$ to denote the graph $G[V(G)\setminus U]$.
Let $S\subseteq E(G)$ for a graph $G$. By $G-S$ we denote by $G-S$ the graph obtained by the deletion of the edges of $S$.
We write $G-e$ instead of $G-\{e\}$ for an edge $e$.
For a vertex $v$, we denote by $N_G(v)$ the \emph{(open) neighborhood} of $v$, i.e., the set of vertices that are adjacent to $v$ in $G$.
Two nonadjacent vertices $u$ and $v$ are \emph{(false) twins} if $N_G(u)=N_G(v)$.
A set of edges with pairwise distinct end-vertices is called a \emph{matching}. A matching $M$ is \emph{induced} if the end-vertices of $M$ are pairwise nonadjacent. A vertex $v$ is \emph{saturated} in a matching $M$ if $v$ is incident to an edge of $M$.
We say that the disjoint union of copies of $K_2$ is a \emph{matching graph}. 
A graph class $\mathcal{C}$ is said to be \emph{hereditary} if for every $G\in\mathcal{C}$ and every induced subgraph $H$ of $G$, $H\in \mathcal{C}$. 
A set of vertices $X\subseteq V(G)$ is a \emph{vertex cover} of a graph $G$ if every edge of $G$ has at least one of its end-vertices in $X$. The minimum size of a vertex cover is called the \emph{vertex-cover number}  of $G$ and is denoted by  $\beta(G)$. 

Let $k$ be a positive integer. A graph $G$ is \emph{(edge) $k$-connected} if for every $S\subseteq E(G)$ with $|S|\leq k-1$, $G-S$ is connected. Since we consider only edge connectivity, whenever we say that a graph $G$ is $k$-connected, we mean that $G$ is edge $k$-connected.
We assume that every graph is $0$-connected. A set of edges $S\subseteq E(G)$ of a connected graph $G$ is an \emph{edge separator} if $G-S$ is disconnected.
An edge $e$ of a connected graph $G$ is a \emph{bridge} if $\{e\}$ is a separator. Clearly, a connected graph is 2-connected if and only if it has no bridge. Let $B$ be the set of bridges of a connected graph $G$. We call a component of $G-B$ a \emph{biconnected component} of $G$. In other words, a biconnected component is an inclusion-wise maximal induced 2-connected subgraph of $G$. 
We say that a biconnected component $L$ of a graph $G$ is a \emph{pendant biconnected component} (or simply a \emph{pendant}) if a unique bridge of $G$ is incident to $V(L)$. A biconnected component is \emph{trivial} if it has a single vertex.
For a graph $G$, we denote by $c(G)$ the number a components of $G$, and for a connected graph $G$, $p(G)$ is the number of pendants. We also denote by $i(G)$ the number of isolated vertices of $G$.

Let $S$ be an inclusion-wise minimal edge separator of a connected graph $G$. Then $G-S$ has exactly two components $C_1$ and $C_2$.
Let $G$ be a spanning subgraph of $F$. We say that an edge $e\in E(F)\setminus E(G)$ \emph{covers} a minimal separator  $S$ of $G$ if $e$ has its end-vertices  in $C_1$ and $C_2$.
The following observation about separators is useful. 

\begin{observation}\label{obs:cover-bridge}
Let $k\geq 2$ be an integer and let a $(k-1)$-connected  graph $G$ be a spanning subgraph of $F$. Then $F$ is $k$-connected if and only if for each edge separator $S$ of $G$ with  $|S|=k-1$, $F$ has an edge that covers it. 
\end{observation} 

We also need some additional terminology and folklore observations for the augmentation of a connected graph to a 2-connected graph.
Let $G$ be a connected graph and let $x$ and $y$ be distinct vertices of $G$. We say that a bridge $uv$ of $G$ \emph{belongs} to an $(x,y)$-path $P$ if $uv\in E(P)$. Similarly, a biconnected component $Q$ is \emph{crossed} by $P$ if $V(Q)\cap V(P)\neq \emptyset$. The following observation show that the choice of an $(x,y)$-path is irrelevant if the biconnected components containing the  end-vertices are given.

\begin{observation}\label{obs:lie}
Let distinct $\{x_1,y_1\}$ and $\{x_1,y_2\}$ be pairs of distinct vertices of a connected graph $G$ such that $x_1,x_2$ are in the same biconnected component of $G$ and, similarly, $y_1,y_2$ are in the same biconnected component of $G$. Let also $P_1$ and $P_2$ be $(x_1,y_1)$ and $(x_2,y_2)$-paths respectively. Then the following holds:
\begin{itemize} 
\item a bridge $uv$ of $G$ belongs to $P_1$ if and only if $uv$ belongs to $P_2$,
\item a biconnected component $Q$  is crossed by $P_1$ if and only if $Q$ is crossed by $P_2$.
\end{itemize}
\end{observation}

\begin{observation}\label{obs:edge-add}
Let $u$ and $v$ be distinct nonadjacent vertices of a connected graph $G$ and let $F$ be a graph obtained from $G$ by the addition of the edge $uv$.  Then $uv$ covers all bridges that belongs to a $(u,v)$-path $P$ in $G$, and for the biconnected components $Q_1,\ldots,Q_s$ that are crossed by $P$, $F[V(Q_1)\cup\ldots\cup V(Q_s)]$ is a biconnected component of $F$.
\end{observation}

In the remaining part of the paper, we will be always assuming that in the instance of the structured augmentation problem, we  have 
\begin{itemize}
\item[($i$)]   $|V(H)|\leq |V(G)|$;
\item[($ii$)]   Graph $H$ has no isolated vertices. 
\end{itemize}
Indeed,   if $|V(H)|>|V(G)|$, then there is no superposition of $G$ and $H$, and thus such an instance is a no-instance. 
For ($ii$),   it is sufficient to observe that mapping of isolated vertices of $H$ to vertices of $G$ does not influence the connectivity of the superposition. 
Another technical detail  should be mentioned here. In  Theorems~\ref{thm:wsca-one} and ~\ref{thm:wsca-two}, we evaluate the running times of algorithms as a function of 
$|V(G)|$ and the vertex cover number of $H$. In order to do this, we should be able to recognize within this time the (trivial) no-instances, where   
 $|V(H)|>|V(G)|$.
 We can verify this condition in time $|V(G)|^{\Oh(1)}$ just by refuting the instances of size more than $|V(G)|^{\Oh(1)}$ after reading the first  $|V(G)|^{\Oh(1)}$ bits.

\section{Augmenting by graphs with  small vertex cover.}\label{sec:weighted} 
In this section we consider the situation when graph $H$ is from a graph class $\mathcal{C}$ with bounded vertex-cover  number. 
 In Subsection~\ref{sec:vc} we show that  in this case 
 \probStrucAugm
and
\probStrucAugmtwo
 are solvable in polynomial time. 
  In Subsection~\ref{sec:hard} we show that this condition is tight by proving that 
  for any hereditary graph class  $\mathcal{C}$ with  unbounded vertex-cover  number, both problems are NP-hard.

\subsection{Algorithms}\label{sec:vc}
We start with a solution  for  \probStrucAugm, which is simpler than the
solution for
\probStrucAugmtwo.

\medskip\noindent\textbf{\probStrucAugm.}
 We need the following lemma.

\begin{lemma}\label{lem:conn-vc}
Let $G$ and $H$ be graphs and let $\varphi\colon V(H)\rightarrow V(G)$ be an injection such that  $F=G\oplus_\varphi H$ is connected. Let also  $X$ be a vertex cover of $H$ of size  $t$.
Then there is a set $Y\subseteq V(H)\setminus X$ of size at most $2(t-1)$ such that for graph $H'=H[X\cup Y]$ and mapping $\psi=\varphi|_{X\cup Y}$, the vertices of $\psi(X\cup Y)$ are in the same connected component of $F'=G\oplus_{\psi}H'$.
\end{lemma}

\begin{proof}
If $|X|=1$, then the claim of the lemma is trivial. Assume that $|X|\geq 2$.

Let $X'\subseteq X$ be an inclusion-wise maximal set such that there  is a set $Y'\subseteq V(H)\setminus X'$ of size at most $2(|X'|-1)$ such that for $H'=H[X'\cup Y']$ and $\psi'=\varphi|_{X'\cup Y'}$, the vertices of $\psi'(X'\cup Y')$ are in the same component of $F'=G\oplus_{\psi'}H'$. Notice that every one-element subset of $X$ satisfies this property and  therefore such a set $X'$ exists. If $X'=X$, then the claim of the lemma holds. Suppose that $X'\subset X$. Let  $s=|X'|<t$.
We show that in this case we can extend $X'$ which will contradict its maximality. 

More precisely, 
we claim that there is $x\in X\setminus X'$ such that for $X''=X'\cup\{x\}$, there  is a set $Y'\subseteq Y''\subseteq V(H)\setminus X''$ of size at most $2s$ such that for $H''=H[X''\cup Y'']$ and $\psi''=\varphi|_{X''\cup Y''}$, the vertices of $\psi''(X''\cup Y'')$ are in the same component of $F''=G\oplus_{\psi''}H''$.

If there is $x\in X\setminus X'$ such that $x$ is in the same component of $F'$ with the vertices of $\psi'(X'\cup Y')$, then the claim holds for $X''=X'\cup\{x\}$ and $Y''=Y'$. Suppose that it is not so, that is, for every  $x\in X\setminus X'$,  $x$ does not belong to  the component of $F'$ with the vertices of $\psi'(X'\cup Y')$.
Recall that $F$ is connected. We select $x\in X\setminus X'$ and a path $P$ joining $\varphi(x)$ and a vertex of $\varphi(X')$ in $F$ in such a way that $P$ contains the minimum number of vertices of $\varphi(X)$. 

Let $P$ be a $(\varphi(x),\varphi(x'))$-path for $x'\in X'$. Notice that $P$ has no internal vertex in $\varphi(X)$. Otherwise (if there is such a vertex $v$) then either $\varphi^{-1}(v)\in X'$,  or $\varphi^{-1}(v)\in X\setminus X'$. In the first case the $(\varphi(x),v)$-subpath of $P$ connects $x$ with a vertex of $\varphi(X')$, and in the second case, 
the $(v,\varphi(x'))$-subpath of $P$ connects a vertex of $\varphi(X\setminus X')$ with $x'\in \varphi(X')$. In both cases this  contradicts the choice of $P$.
We obtain that $V(P)\cap \varphi(X)=\{x,x'\}$. 

This implies that $P$ contains at most 2 edges that are not edges of $G$. Moreover, because $X$ is a vertex cover of $H$, 
 every such edge is incident either with $\varphi(x)$, or with $\varphi(x')$. 
Denote by $S$ the set of endpoints of these edges distinct from $\varphi(x)$ and $\varphi(x')$. 
We have that $S\subseteq\varphi(V(H)\setminus X)$ and $|S|\leq 2$. Let $X''=X\cup\{x\}$ and $Y''=Y'\cup S$. We obtain that $Y'\subseteq Y''\subseteq V(H)\setminus X''$ and $|Y''|\leq|Y'|+2\leq 2s$. Let $\psi''=\varphi|_{X\cup Y''}$. Observe that $P$ is a path in $F''=G\oplus_{\psi''}H''$. It implies that the vertices of $\psi''(X''\cup Y'')$ are in the same component of $F''$.

We obtain a contradiction that proves that $X'=X$ and the lemma holds.
\end{proof}

 Let us remind, that, given a positive integer $t$, a graph class
$\mathcal{C}$  has    vertex-cover  number at most $t$ if every graph $H\in\mathcal{C}$ has a vertex cover of size at most $t$. 
We are ready to prove the main theorem about \probStrucAugm.

\begin{theorem}\label{thm:wsca-one}
Let $t$ be a positive integer and $\mathcal{C}$ be a graph class of vertex-cover  number at most $t$. Then for any $H\in \mathcal{C}$, 
\probStrucAugm is  solvable in time $|V(G)|^{\Oh(t)}\cdot \log W$.
\end{theorem}

\begin{proof} 
Let $G$ and $H\in \mathcal{C}$ be graphs and let $\cost\colon \binom{V(G)}{2}\rightarrow \mathbb{N}_0 $ be a weight function. We show that we can find in time $|V(G)|^{\Oh(t)}\cdot \log W$ an injective  mapping $\varphi\colon V(H)\rightarrow V(G)$ such that $F=G\oplus_{\varphi}H$ is connected and $\cost(\varphi)=\sum_{xy\in E(H)}\cost(\varphi(x)\varphi(y))$ is minimum if $\varphi$ exists.

  Let us remind that without loss of generality,  
 we can assume that $|V(H)|\leq |V(G)|$ and $H$ has no isolated vertices. 

We start from finding a vertex cover $X$ of  size at most $t$ in $H$. Since we aim for an algorithm with running  time $|V(G)|^{\Oh(t)}\cdot \log W$,  vertex cover $X$ can be found by brute-force checking of all subsets of $V(H)$ of size at most $t$. If we fail to find $X$ of size at most $t$, it means that $H\not\in \mathcal{C}$, in this case 
we return the answer NO and stop. Assume that $X$ exists. 

Suppose that there is an injective  mapping $\varphi\colon V(H)\rightarrow V(G)$ such that $F=G\oplus_{\varphi}H$ is connected and assume that for $\varphi$,  $\cost(\varphi)$ is minimum. By Lemma~\ref{lem:conn-vc}, there is a set $Y\subseteq V(H)\setminus X$ of size at most $2(t-1)$ such that for $H'=H[X\cup Y]$ and $\psi=\varphi|_{X\cup Y}$, the vertices of $\psi(X\cup Y)$ are in the same component of $F'=G\oplus_{\psi}H'$.
Considering all possibilities, we guess $Y$ in time $|V(H)|^{\Oh(t)}$. 

Now we consider all possible injective mapping $\psi\colon X\cup Y\rightarrow V(G)$ such that the vertices of $\psi(X\cup Y)$ are in the same connected component of 
$F'=G\oplus_\psi H'$,  where $H'=H[X\cup Y]$. Notice that there are at most $|V(G)|^{3t-2}$ such mappings that can be generated in time $|V(G)|^{\Oh(t)}$. If we fail to find $\psi$, we reject the current choice of $Y$. 
Otherwise, for every $\psi$, we try to extend it to an injection $\varphi\colon V(H)\rightarrow V(G)$ such that $F=G\oplus_{\varphi}H$ is connected, and among all extensions we choose one that provides the minimum weight $\cost(\varphi)$.

Let $Z=V(H)\setminus(X\cup Y)$. 
The vertices of $\psi(X\cup Y)$ are in the same component of $F'$. Denote this component by $F_0$ and denote by $F_1,\ldots,F_r$ the other components of this graphs. Recall that $Z$ is an independent set of $H$ and each vertex of $Z$ has an incident edge with one endpoint in $X$. It follows that for an injection $\varphi\colon V(H)\rightarrow V(G)$ such that $\psi=\varphi|_{X\cup Y}$, $F=G\oplus_\varphi H$ is connected if and only if for every $i\in\{1,\ldots,r\}$, there is $v\in V(F_i)$ such that $v\in \varphi(Z)$. Hence, if $r>|Z|$, we cannot extend $\psi$. In this case we discard the current choice of $\psi$. 

Assume from now  that $Y$ and $\psi$ are fixed, $F'=G\oplus_\psi H'$ is connected and $r\leq|Z|$. For $z\in Z$ and $v\in V(G)\setminus \psi(X\cup Y)$, we define the weight of mapping $z$ to $v$ as 
\[w(z,v)=\sum_{u\in N_G(v)\cap \psi(N_H(z))}\cost(uv),\] 
that is, $w(z,x)$ is the weight of edges that is added to the weight of mapping if we decide to extend $\psi$ by mapping $z$ to $v$.  Let $W=\max\{w(z,v)\mid z\in Z,v\in V(G)\setminus \psi(X\cup Y)\}+1$. 
We construct the weighted auxiliary bipartite graph $\mathcal{G}$ with the bipartition $(A,B)$ of its vertex set and the weight function $f\colon E(\mathcal{G})\rightarrow \mathbb{N}_0$ as follows.
\begin{itemize}
\item Set $A=(V(F_0)\setminus \psi(X\cup Y))\cup V(F_1)\cup\ldots\cup V(F_r)=V(G)\setminus \psi(X\cup Y)$.
\item Construct a set of vertices $S_0$ of size $|V(F_0)|-|X\cup Y|$ and sets $S_i$ of size $|V(F_i)|-1$ for $i\in\{1,\ldots,r\}$.
\item Set $B=Z\cup S_0\cup\ldots\cup S_r$.
\item For each $z\in Z$ and $v\in A$, construct an edge $zv$ and set $f(zv)=w(z,v)$.
\item For each $u\in S_0$ and $v\in V(F_0)\setminus \psi(X\cup Y)$, construct an edge $uv$ and set $f(uv)=W$.
\item For each $\in\{1,\ldots,r\}$, do the following: for each $u\in S_i$ and $v\in V(F_i)$,  construct an edge $uv$ and set $f(uv)=W$. 
\end{itemize} 
We find a matching $M$ in $\mathcal{G}$ that saturates every vertex of $A$ and has the minimum weight using the Hungarian algorithm~\cite{FredmanT87,Kuhn55} in time $\Oh(|V(G)|^3\cdot\log W)$.
  
Observe that a matching that saturates every vertex of $A$ exists, because $r\leq Z$. We can construct such a matching  by selecting one vertex in $V(F_i)$ for each $i\in\{1,\ldots,r\}$ and matching it with a vertex of $Z$. Then we complement this set of edges to a matching saturating $A$ by adding edges incident to $S_0\cup\ldots\cup S_r$. For the matching $M$ that has minimum weight, we can also observe the following.

First, note that 
\begin{eqnarray} \label{eqn_i}
\mbox{every vertex of } Z \mbox{ is saturated by } M.
\end{eqnarray}
Indeed, targeting towards  a contradiction, assume that $z\in Z$ is not saturated. Since $|V(H)|\leq|V(G)|$, there is $uv\in M$ such that $u\in S_0\cup\ldots\cup S_r$ and $v\in A$. We replace $uv$ by $zv$ in $M$. Because $f(uv)=W>w(zv)$, we obtain a matching with a smaller weight. This  contradicts the choice of $M$.

Next, we claim that 
\begin{eqnarray} \label{eqn_ii}
\mbox{there is } zv\in M \mbox{ such that } z\in Z \mbox{ and }v\in V(F_i).
\end{eqnarray}
 Indeed, this is because   the vertices of $V(F_i)$ are adjacent to $|V(F_i)|-1$ vertices of $S_i$ and all other their neighbors are in $Z$. 

Finally, we have that among all matching saturating $A$, $M$ is a matching satisfying~\eqref{eqn_i} and \eqref{eqn_ii} such that for $M'=\{zv\in M\mid z\in Z\}$, $f(M')$ is minimum. To see it, observe that  $f(uv)=W$ for $uv\in M\setminus M'$. Hence, $f(M\setminus M')=(|A|-|Z|)W$, because $|M\setminus M'|=|A|-|Z|$ by (\ref{eqn_i}).
Therefore, 
$f(M')=f(M)-f(M\setminus M')=f(M)-(|A|-|Z|)W$.

For every $z\in Z$, we define $\varphi(z)=v$, where $zv\in M'$ and $\varphi(x)=\psi(x)$ for $x\in X\cup Y$. Clearly, $\varphi$ is an extension of $\psi$.
 By (\ref{eqn_i}), $\varphi$ is an injective mapping of $V(H)$ to $V(G)$. By (\ref{eqn_ii}) and the choice of $X$ and $Y$, we obtain that $G\oplus_\varphi H$ is connected.  
We claim that $\varphi$ is an extension of $\psi$ such that $F=G\oplus_{\varphi}H$ is connected that has the minimum total weight $\cost(\varphi)=\sum_{xy\in E(H)}\cost(\varphi(x)\varphi(y))$.

Recall that by the definition of the weight function $f$, $f(zv)=w(z,v)$ for $z\in Z$ and $v\in A$, and 
$w(z,v)=\sum_{u\in N_G(v)\cap \psi(N_H(z))}\cost(uv)$
in this case. Let $R= \sum_{xy\in E(H),~x,y\in X\cup Y}\cost(\psi(x)\psi(y))$.
It follows that 
\begin{align}\label{eq:w-phi}
\cost(\varphi)=&\sum_{xy\in E(H)}\cost(\varphi(x)\varphi(y))=\sum_{xy\in E(H),~x,y\in X\cup Y}\cost(\psi(x)\psi(y))+\sum_{xz\in E(H),x\in X,z\in Z}\cost(\psi(x)\varphi(z))\nonumber\\
=&R+\sum_{zv\in M'}w(x,z)=R+f(M').
\end{align}
Suppose that
$\varphi'\colon V(H)\rightarrow V(G)$ is an injection that extends $\psi$ such that  $F'=G\oplus_{\varphi'}H$ is connected.
We construct the matching $\tilde{M}$ in $\mathcal{G}$ as follows. For every $z\in Z$, we include $z\varphi'(z)$ in $\tilde{M}$. Denote by $\tilde{M}'$ the obtained matching.
Notice that every vertex of $Z$ is saturated in $\tilde{M}'$ and, therefore, $A$ has $|Z|$ saturated in $\tilde{M}'$ vertices. Hence, $\tilde{M}'$ satisfies (\ref{eqn_i}).
Since  $F'$ is connected, at least one vertex of $V(F_i)$ is saturated for $i\in\{1,\ldots,r\}$ and, therefore, $\tilde{M}'$ satisfies (\ref{eqn_ii}). Then we complement $\tilde{M}'$ to $\tilde{M}$: 
for every nonsaturated vertex $v\in A$, we arbitrarily pick a nonsaturated neighbor $u\in B\setminus Z$ and include $vu$ in $\tilde{M}$.  
This choice is possible, because  $|S_0|=|V(F_0)|-|X\cup Y|$ and $|S_i|=|V(F_i)|-1$ for $i\in\{1,\ldots,r\}$. Since $\tilde{M}'$ and, therefore, $\tilde{M}$ satisfies (\ref{eqn_i}) and (\ref{eqn_ii}), we obtain that $f(M)\leq f(\tilde{M})$ and $f(M')\leq f(\tilde{M}')$. In the same way as in (\ref{eq:w-phi}), we 
 have that $\cost(\varphi')=R+ f(\tilde{M}')$. 
Then $\cost(\varphi') \geq \cost(\varphi)$ by (\ref{eq:w-phi})
and this proves the claim.

Recall that we try all possible choices of $Y$ and for every choice of $Y$, we consider all possible choices of $\psi$. If we fail to find an injection $\varphi\colon V(H)\rightarrow V(G)$ such that $\varphi$ is an extension of $\psi$ and $F=G\oplus_{\varphi}H$ is connected we return the answer NO. Otherwise, we return  $\varphi$ that provides the minimum weight.

To complete the proof, observe that the total running time of the algorithm is $|V(G)|^{\Oh(t)}\cdot \log W$.
\end{proof}

\medskip\noindent\textbf{\probStrucAugmtwo.} The algorithm for \probStrucAugmtwo is more technical. We start with a lemma, which is similar to 
 Lemma~\ref{lem:conn-vc}.

\begin{lemma}\label{lem:two-conn-vc}
Let $G$ and $H$ be graphs such that $G$ is connected, and let $\varphi\colon V(H)\rightarrow V(G)$ be an injection such that  $F=G\oplus_\varphi H$ is connected. Suppose that $X$ is a vertex cover of $H$ and $t=|X|$.
Then there is a set $Y\subseteq V(H)\setminus X$ of size at most $2(t-1)$ such that for $H'=H[X\cup Y]$ and $\psi=\varphi|_{X\cup Y}$, the vertices of $\psi(X\cup Y)$ are in the same biconnected component of $F'=G\oplus_{\psi}H'$.
\end{lemma}

\begin{proof}
For $|X|=1$   lemma is trivial, so we assume that $|X|\geq 2$.

Let $X'\subseteq X$ be an inclusion-wise maximal set among all sets with the following property:  there  is a set $Y'\subseteq V(H)\setminus X'$ of size at most $2(|X'|-1)$ such that for $H'=H[X'\cup Y']$ and $\psi'=\varphi|_{X'\cup Y'}$, the vertices of $\psi'(X'\cup Y')$ are in the same biconnected component of $F'=G\oplus_{\psi'}H'$. Since  every one-element subset of $X$ satisfies this property such a set $X'$ exists.

In order to prove the lemma, we prove that $X'=X$.

   Targeting towards a contradiction,  
suppose that $X'$ is a proper subset of $X$. Let $s<t$ be the size of $X'$.
We show that then we can extend $X'$ contradicting its maximality. 
More precisely, 
we claim that there is $x\in X\setminus X'$ such that for $X''=X'\cup\{x\}$, there  is a set $Y'\subseteq Y''\subseteq V(H)\setminus X''$ of size at most $2s$ such that for $H''=H[X''\cup Y'']$ and $\psi''=\varphi|_{X''\cup Y''}$, the vertices of $\psi''(X''\cup Y'')$ are in the same biconnected component of $F''=G\oplus_{\psi''}H''$.

If there is $x\in X\setminus X'$ such that $x$ is in the same biconnected component of $F'$ with the vertices of $\psi'(X'\cup Y')$, then the claim holds for $X''=X'\cup\{x\}$ and $Y''=Y'$. Suppose that it is not so, that is, for every  $x\in X\setminus X'$,  $x$ does not belong to  the biconnected component of $F'$ with the vertices of $\psi'(X'\cup Y')$.

Recall that $G$ is connected. Therefore, $F'$ is connected as well. Since the vertices of $\varphi(X)$ do not belong to the same biconnected component, $F'$ is not 2-connected.
Let $B$ be the set of bridges of $F'$.

Suppose that there is an edge $x'y\in E(H)$ with $x'\in X'$ such that there is a biconnected component $Q$ of $F'$ that is crossed by  a $(\varphi(x'),\varphi(y))$-path $P$ in $F'$ and $Q$ contains a vertex $v\in\varphi(X\setminus X')$. Let $x=\varphi^{-1}(x)$.
Consider $X''=X'\cup\{x\}$ and $Y''=Y'\cup\{y\}$. Clearly, $Y'\subseteq Y''\subseteq V(H)\setminus X''$ and $|Y''|\leq 2s$.
Let $H''=H[X\cup Y'']$ and $\psi''=\varphi|_{X\cup Y''}$. Then by Observation~\ref{obs:edge-add}, the vertices of $\psi''(X''\cup Y'')$ are in the same biconnected component of $F''=G\oplus_{\psi''}H''$. This contradicts the choice of $X'$. 

Suppose now that  there is an edge $xy\in E(H)$ with $x\in X\setminus X'$ such that the biconnected component $Q$ of $F'$ that contains the vertices of $\varphi(X'\cup Y')$ is crossed by a $(\varphi(x),\varphi(y))$-path $P$ in $F'$. 
Consider $X''=X'\cup\{x\}$ and $Y''=Y'\cup\{y\}$. We have that $Y'\subseteq Y''\subseteq V(H)\setminus X''$ and $|Y''|\leq 2s$.
Let $H''=H[X\cup Y'']$ and $\psi''=\varphi|_{X\cup Y''}$. Then again  by Observation~\ref{obs:edge-add}, we have that  the vertices of $\psi''(X''\cup Y'')$ are in the same biconnected component of $F''=G\oplus_{\psi''}H''$. Again, this contradicts the choice of $X'$. 

Now we assume that the two previous cases do not hold. 
In particular, in this situation,   not all bridges of $F'$ are covered by edges $pq\in E(F)\setminus E(F')$ with $\varphi^{-1}(p)\in X'$ or $\varphi^{-1}(q)\in X'$ and not all bridges of $F'$ are covered by edges $pq\in E(F)\setminus E(F')$ with $\varphi^{-1}(p)\in X\setminus X'$ or $\varphi^{-1}(q)\in X\setminus X'$. Since $F$ is 2-connected, by Observation~\ref{obs:cover-bridge} all bridges of $G$ should be covered by edges of $F$. Hence,
there are distinct $uv,u'v'\in B$ such that $u,u'\in V(Q)$ for some biconnected component $Q$ of $F'$, $uv$ is covered by 
$pq\in E(F)\setminus E(F')$ with $\varphi^{-1}(p)\in X\setminus X'$ and $u'v'$ is covered by $p'q'\in E(F)\setminus E(F')$ with $\varphi^{-1}(p')\in X'$.
Let $x=\varphi{-1}(u)$, $y=\varphi^{-1}(v)$ and $y'=\varphi^{-1}(v')$. Consider $X''=X'\cup\{x\}$ and $Y''=Y'\cup\{y,y'\}$. Clearly, $Y'\subseteq Y''\subseteq V(H)\setminus X''$ and $|Y''|\leq 2s$.
Let $H''=H[X\cup Y'']$ and $\psi''=\varphi|_{X\cup Y''}$. By Observation~\ref{obs:edge-add}, the vertices of $\psi''(X''\cup Y'')$ are in the same biconnected component of $F''=G\oplus_{\psi''}H''$, which, again, this contradicts the choice of $X'$.

Hence $X'=X$ and  the lemma holds.
\end{proof}

\begin{theorem}\label{thm:wsca-two}
Let $t$ be a positive integer and $\mathcal{C}$ be a graph class of vertex-cover  number at most $t$. Then for any $H\in \mathcal{C}$, 
 \probStrucAugmtwo is  solvable in time $|V(G)|^{\Oh(2^t)}\log W$.
\end{theorem}

\begin{proof}
Let $G$ and $H$ be graphs such that $G$ is connected and $H\in \mathcal{C}$. Let $\cost\colon \binom{V(G)}{2}\rightarrow \mathbb{N}_0 $ be a weight function. Similarly to the proof of Theorem~\ref{thm:wsca-one} we show that we can find in time $|V(G)|^{\Oh(2^t)}\cdot\log W$ the minimum value of $\cost(\varphi)=\sum_{xy\in E(H)}\cost(\varphi(x)\varphi(y))$ for an injective  mapping $\varphi\colon V(H)\rightarrow V(G)$ such that $F=G\oplus_{\varphi}H$ is connected if such a mapping $\varphi$ exists.  

The first steps of our algorithm are the same as in the proof of Theorem~\ref{thm:wsca-one}.
  Again, we remind that   $|V(H)|\leq |V(G)|$ and that $H$ has no isolated vertices. 

Next, we find a vertex cover $X$ of minimum size in $H$   of size at most $t$ in time $|V(G)|^{\Oh(t)}$. If we fail to find $X$ of size at most $t$, then $H\not \in \mathcal{C}$. We return  NO and stop. From now on we assume that $X$ exists. 

Suppose that there is an injective  mapping $\varphi\colon V(H)\rightarrow V(G)$ such that $F=G\oplus_{\varphi}H$ is 2-connected and assume that for $\varphi$,  $\cost(\varphi)$ is minimum. By Lemma~\ref{lem:two-conn-vc}, there is a set $Y\subseteq V(H)\setminus X$ of size at most $2(t-1)$ such that for $H'=H[X\cup Y]$ and $\psi=\varphi|_{X\cup Y}$, the vertices of $\psi(X\cup Y)$ are in the same biconnected component of $F'=G\oplus_{\psi}H'$.
Considering all possibilities, we guess $Y$ in time $|V(H)|^{\Oh(t)}$.

Now we consider all possible injective mapping $\psi\colon X\cup Y\rightarrow V(G)$ such that the vertices of $\psi(X\cup Y)$ are in the same biconnected component of 
$F'=G\oplus_\psi H'$ where $H'=H[X\cup Y]$. Notice that there at most $|V(G)|^{3t-2}$ such mappings that can be generated in time $|V(G)|^{\Oh(t)}$. If we fail to find $\psi$, we reject the current choice of $Y$. 
Otherwise, for every $\psi$, we try to extend it to an injection $\varphi\colon V(H)\rightarrow V(G)$ such that $F=G\oplus_{\varphi}H$ is 2-connected, and among all extensions we choose one that provides the minimum weight $\cost(\varphi)$.

Let $Z=V(H)\setminus(X\cup Y)$. 
The vertices of $\psi(X\cup Y)$ are in the same biconnected component of $F'$. Denote this biconnected component by $F_0$ and denote by $F_1,\ldots,F_r$ the pendant biconnected components of $F'$ that are distinct from $F_0$. Recall that $Z$ is an independent set of $H$ and each vertex of $Z$ has an incident edge with one endpoint in $X$. 
By Observation~\ref{obs:cover-bridge}, we obtain the following crucial property.

 For an injection $\varphi\colon V(H)\rightarrow V(G)$ such that $\psi=\varphi|_{X\cup Y}$, $F=G\oplus_\varphi H$ is 2-connected if and only if
\begin{itemize}
\item[(i)] for every $i\in\{1,\ldots,r\}$, there is $v\in V(F_i)$ such that $v\in \varphi(Z)$, and
\item[(ii)] if $v$ is the unique element of $V(F_i)\cap \varphi(Z)$ and $v$ is incident to a bridge $vu$ of $G$, then there is $x\in X$ such that $\varphi(x)\neq u$ and $x$ is adjacent to 
$\varphi^{-1}(v)$ in $H$.
\end{itemize}

Similarly to the proof of Theorem~\ref{thm:wsca-one}, we solve auxiliary matching problems to find the minimum weight of $\varphi$ but now, due the condition (ii), the algorithm becomes more complicated and we are using dynamic programming.

For $z\in Z$ and $v\in V(G)\setminus \psi(X\cup Y)$, we define the weight of mapping $z$ to $v$ as 
\begin{equation}\label{eq:w-two}
w(z,v)=\sum_{u\in N_G(v)\cap \psi(N_H(z))}\cost(uv),
\end{equation}
that is, $w(z,x)$ is the weight of edges that is added to the weight of mapping if we decide to extend $\psi$ by mapping $z$ to $v$.  Our aim is to find the extension $\varphi$ of $\psi$ that satisfies (i) and (ii) such that the total weight of the mapping of the vertices of $Z$ to verices of $V(G)\setminus \psi(X\cup Y)$ by $\varphi$ is minimum.

Since $X$ is a vertex cover of $H$ of size $t$, the set $Z$ can be partitioned into $s\leq 2^t$ classes of false twins $Z_1,\ldots,Z_s$. 
Let $p_i=|Z_i|$ for $i\in \{1,\ldots,s\}$. We exploit the following property of false twins in $Z$: if $x,y\in Z_i$, then $w(x,v)=w(y,v)$ for $v\in V(G)\setminus \psi(X\cup Y)$.

For each $s$-tuple of integers $(q_1,\ldots,q_s)$ such that $0\leq q_i\leq p_i$,  for $i\in \{1,\ldots,s\}$ and each $h\in \{0,\ldots,r\}$, we define 
\begin{equation}\label{eq:dp-funct}
\alpha_h(q_1,\ldots,q_s)=\min_{\xi}\sum_{z\in Z'} w(z,\xi(z)),
\end{equation}
where $Z'\subseteq Z$ such that $|Z'\cap Z_i|=q_i$ for $i\in\{1,\ldots,s\}$ and the minimum is taken over all injective mappings $\xi\colon Z'\rightarrow (V(F_0)\setminus \psi(X\cup Y))\cup V(F_1)\cup\ldots\cup V(F_h)$ such that the following conditions are satisfied:
\begin{itemize}
\item[(a)]  for every $i\in\{1,\ldots,h\}$, there is $v\in V(F_i)$ such that $v\in \xi(Z')$, and
\item[(b)] if $v$ is a unique element of $V(F_i)\cap \xi(Z')$ for some $i\in\{1,\ldots,h\}$
and $v$ is incident to a bridge $vu$ of $G$, then there is $x\in X$ such that $\psi(x)\neq u$ and $x$ is adjacent to 
$\xi^{-1}(v)$ in $H$.
\end{itemize}
If such a mapping $\xi$ does not exist, then we assume that $\alpha_h(q_1,\ldots,q_s)=+\infty$.
Recall that if $x,y\in Z_i$, then $w(x,v)=w(y,v)$ for $v\in V(G)\setminus \psi(X\cup Y)$. It implies that the function $\alpha_h(q_1,\ldots,q_s)$ depends only on the values of $q_1,\ldots,q_s$.

We claim that computing $\alpha_r(p_1,\ldots,p_s)$ is equivalent to finding an extension $\varphi$ of $\psi$ of minimum weight  such that $F=G\oplus_\varphi H$ is 2-connected. 

Assume that  $\alpha_r(p_1,\ldots,p_s)<+\infty$.
Notice that $Z'=Z$ if $q_i=p_i$ for $i\in \{1,\ldots,s\}$.
Let  $\xi\colon Z\rightarrow (V(F_0)\setminus \psi(X\cup Y))\cup V(F_1)\cup\ldots\cup V(F_h)$ be an injection that  provides the minimum in (\ref{eq:dp-funct}), that is,
$\alpha_r(p_1,\ldots,p_s)=\sum_{z\in Z} w(z\xi(z))$. 
We define $\varphi(z)=\xi(z)$ for $z\in Z$ and $\varphi(x)=\psi(x)$ for $x\in X\cup Y$. Clearly, $\varphi$ is an extension of $\psi$. Because $\xi$ is an injection, we have that $\varphi$ is an injective mapping. Since $\xi$ satisfies (a) and (b), we obtain that $\varphi$ satisfies (i) and (ii) and, therefore,
 $F=G\oplus_{\varphi}H$ is 2-connected.
Let $R= \sum_{xy\in E(H),~x,y\in X\cup Y}\cost(\psi(x)\psi(y))$.
Then using (\ref{eq:w-two}), we have that 
\begin{align}\label{eq:w-phi-two}
\cost(\varphi)=&\sum_{xy\in E(H)}\cost(\varphi(x)\varphi(y))=\sum_{xy\in E(H),~x,y\in X\cup Y}\cost(\varphi(x)\varphi(y))+\sum_{xz\in E(H),~x\in X,z\in Z}\cost(\varphi(x)\varphi(y))=\nonumber\\
=&R+\sum_{z\in Z}w(z,\varphi(z))=R+\sum_{z\in Z}w(z,\xi(z))=R+\alpha_r(p_1,\ldots,p_s).
\end{align}
Let 
$\varphi'\colon V(H)\rightarrow V(G)$ be an injection that extends $\psi$ such that  $F'=G\oplus_{\varphi'}H$ is 2-connected. 
We define  $\xi'\colon Z\rightarrow (V(F_0)\setminus \psi(X\cup Y))\cup V(F_1)\cup\ldots\cup V(F_h)$ by setting $\xi'(z)=\varphi'(z)$ for $z\in Z$.
Since $\varphi'$ is an injection, $\xi'$ is also an injection. Because $F'$ is 2-connected, 
$\varphi$ satisfies (i) and (ii). This implies that $\xi'$ satisfies (a) and (b).
Therefore, $\sum_{z\in Z} w(z,\xi'(z))\geq  \alpha_r(p_1,\ldots,p_s)$. 
Similarly to (\ref{eq:w-phi-two}), we have that $\cost(\varphi')=R+\sum_{z\in Z}w(z,\xi'(z))\geq R+\alpha_r(p_1,\ldots,p_s)$. 
We conclude that $\varphi$ is an extension $\varphi$ of $\psi$ of minimum weight  such that $F=G\oplus_\varphi H$ is 2-connected. 

Suppose that $\alpha_r(p_1,\ldots,p_s)=+\infty$. It implies that there is no injection $\xi\colon Z\rightarrow (V(F_0)\setminus \psi(X\cup Y))\cup V(F_1)\cup\ldots\cup V(F_h)$ satisfying (a) and (b).
But this immediately implies that there is no injective extension $\varphi$ of $\psi$ satisfying (i) and (ii). This completes the proof of the claim.

We use dynamic programming to compute $\alpha_h$ consequently for $h=0,1,\ldots,r$.

We start with computing $\alpha_0(q_1,\ldots,q_s)$ for each $s$-tuple $(q_1,\ldots,q_s)$. Notice that the conditions (a) and (b) are irrelevant in this case, because they concern  only $h\geq 1$. 
We construct the auxiliary complete bipartite graph $\mathcal{G}_0$ with the bipartition $(V(F_0)\setminus \psi(X\cup Y),Z')$ of its vertex set  and define the weight of each edge $zv$ for $z\in Z'$ and $v\in V(F_0)\setminus \psi(X\cup Y)$ as $w(z,v)$.
We find a matching $M$ in $\mathcal{G}_0$ that saturates every vertex of $Z'$ and has the minimum weight using the Hungarian algorithm~\cite{FredmanT87,Kuhn55} in time $\Oh(|V(G)|^3\cdot \log W)$.
If there is no matching saturating $Z'$, we set $\alpha_0(q_1,\ldots,q_s)=+\infty$. Otherwise, $\alpha_0(q_1,\ldots,q_s)=w(M)$. It is straightforward to verify the correctness of computing 
$\alpha_0(q_1,\ldots,q_s)$ by the definition of this function.

Assume that $h\geq 1$ and we already computed the table of values of $\alpha_{h-1}(q_1,\ldots,q_s)$. We explain how to construct the table of values of $\alpha_{h-1}(q_1,\ldots,q_s)$.
The the computation is based on the observation that an injective mapping $\xi\colon Z'\rightarrow (V(F_0)\setminus \psi(X\cup Y))\cup V(F_1)\cup\ldots\cup V(F_h)$ can be seen as the union of two injections $\xi'\colon Z''\rightarrow (V(F_0)\setminus \psi(X\cup Y))\cup V(F_1)\cup\ldots\cup V(F_{h-1})$ and $\lambda\colon Z'''\rightarrow V(F_h)$ for the appropriate partition $(Z'',Z''')$ of $Z'$. 

For each $s$-tuple of integers $(q_1,\ldots,q_s)$ such that $0\leq q_i\leq p_i$ for $i\in \{1,\ldots,s\}$, we define 
\begin{equation}\label{eq:dp-funct-prime}
\alpha_h'(q_1,\ldots,q_s)=\min_{\lambda}\sum_{z\in Z'} w(z,\xi(z)),
\end{equation}
where $Z'\subseteq Z$ such that $|Z'\cap Z_i|=q_i$ for $i\in\{1,\ldots,s\}$ and the minimum is taken over all injective mappings $\lambda\colon Z'\rightarrow V(F_h)$ such that the following conditions are fulfilled:
\begin{itemize}
\item[(a$^*$)]  there is $v\in V(F_h)$ such that $v\in \lambda(Z')$, and
\item[(b$^*$)] if $v$ is the unique element of $V(F_h)\cap \lambda(Z')$ and $v$ is incident to a bridge $vu$ of $G$, then there is $x\in X$ such that $\psi(x)\neq u$ and $x$ is adjacent to 
$\lambda^{-1}(v)$ in $H$.
\end{itemize}
If such a mapping $\lambda$ does not exist, then we assume that $\alpha_h'(q_1,\ldots,q_s)=+\infty$.
As for $\alpha_h(q_1,\ldots,q_s)$, $\alpha_h'(q_1,\ldots,q_s)$  depends only on the values of $q_1,\ldots,q_s$, because if $x,y\in Z_i$, then $w(x,v)=w(y,v)$ for $v\in V(G)\setminus \psi(X\cup Y)$. 

Let $uv$ be the unique bridge of $G$ with $v\in V(F_h)$. Suppose that for an $s$-tuple $(q_1,\ldots,q_s)$, we obtain that  
$|Z'|=1$ and for the unique vertex $z\in Z'$, $z$ has a unique neighbor $x\in X$ in $H$ and $\psi(x)=u$. Then we set $\alpha_h'(q_1,\ldots,q_s)=+\infty$ if $|V(F_h)|=1$ and $\alpha_h'(q_1,\ldots,q_s)=\min\{w(zv')\mid v'\in V(F_h)\setminus\{v\}\}$ otherwise. For other $s$-tuples $(q_1,\ldots,q_s)$, we compute $\alpha_h'(q_1,\ldots,q_s)$ as follows.
We construct the auxiliary complete bipartite graph $\mathcal{G}_h$ with the bipartition $(V(F_h),Z')$ of its vertex set  and define the weigh of each edge $zv$ for $z\in Z'$ and $v\in V(F_0)\setminus \psi(X\cup Y)$ as $w(zv)$.
We find a matching $M$ in $\mathcal{G}_h$ that saturates every vertex of $Z'$ and has the minimum weight using the Hungarian algorithm~\cite{FredmanT87,Kuhn55} in time $\Oh(|V(G)|^3\cdot \log W)$.
If there is no matching saturating $Z'$, we set $\alpha_h'(q_1,\ldots,q_s)=+\infty$. Otherwise, $\alpha_h'(q_1,\ldots,q_s)=w(M)$. It is again straightforward to verify the correctness of computing 
$\alpha_h'(q_1,\ldots,q_s)$ using the definition of this function.

Now, to compute $\alpha_{h}(q_1,\ldots,q_s)$, we use the equation:
\begin{equation}\label{eq:final-dp}
\alpha_{h}(q_1,\ldots,q_s)=\min \{\alpha_{h-1}(q_1',\ldots,q_s')+\alpha_{h}'(q_1'',\ldots,q_s'')\},
\end{equation}
where the minimum is taken over all $s$-tuples $(q_1',\ldots,q_s')$ and $(q_1'',\ldots,q_s'')$ such that $q_i=q_i'+q_i''$ for $i\in\{1,\ldots,s\}$.

To show correctness, we prove first that 
\begin{equation}\label{eq:lower}
\alpha_{h}(q_1,\ldots,q_s)\geq \min \{\alpha_{h-1}(q_1',\ldots,q_s')+\alpha_{h}'(q_1'',\ldots,q_s'')\}.
\end{equation}
The inequality is trivial if $\alpha_{h}(q_1,\ldots,q_s)=+\infty$. Assume that $\alpha_{h}(q_1,\ldots,q_s)<+\infty$. Then there is an injective mappings $\xi\colon Z'\rightarrow (V(F_0)\setminus \psi(X\cup Y))\cup V(F_1)\cup\ldots\cup V(F_h)$ satisfying (a) and (b) such that $\alpha_h(q_1,\ldots,q_s)=\sum_{z\in Z'} w(z,\xi(z))$.
Let $Z''=\{z\in Z\mid \xi(z)\in  (V(F_0)\setminus \psi(X\cup Y))\cup V(F_1)\cup\ldots\cup V(F_{h-1})\}$ and  
$Z'''=\{z\in Z\mid \xi(z)\in   V(F_{h})\}$. Denote by $\xi'$ the restriction of $\xi$ on $(V(F_0)\setminus \psi(X\cup Y))\cup V(F_1)\cup\ldots\cup V(F_{h-1})$ and let 
$\lambda=\xi|_{Z'''}$. 

We have that
$\xi'$ is an injective mapping of $Z''$ to $(V(F_0)\setminus \psi(X\cup Y))\cup V(F_1)\cup\ldots\cup V(F_{h-1})$ such that the following holds:
\begin{itemize}
\item[(a$'$)]  for every $i\in\{1,\ldots,h-1\}$, there is $v\in V(F_i)$ such that $v\in \xi'(Z'')$, and
\item[(b$'$)] if $v$ is a unique element of $V(F_i)\cap \xi'(Z')$ for some $i\in\{1,\ldots,h-1\}$
and $v$ is incident to a bridge $vu$ of $G$, then there is $x\in X$ such that $\psi(x)\neq u$ and $x$ is adjacent to 
$\xi'^{-1}(v)$ in $H$.
\end{itemize}
Let $q_i'=|Z''\cap Z_i|$ for $i\in \{1,\ldots,s\}$. By the definition of $\alpha_{h-1}$, we have that 
\begin{equation}\label{eq:alphah-1}
\alpha_{h-1}(q_1',\ldots,q_s')\leq \sum_{z\in Z''} w(z,\xi'(z)).
\end{equation}

Similarly, we obtain that $\lambda$ is an injective mapping of $Z'''$ to $V(F_h)$ such that the following holds:
\begin{itemize}
\item[(a$^{**}$)]  there is $v\in V(F_h)$ such that $v\in \lambda(Z''')$, and
\item[(b$^{**}$)] if $v$ is the unique element of $V(F_h)\cap \lambda(Z''')$ and $v$ is incident to a bridge $vu$ of $G$, then there is $x\in X$ such that $\psi(x)\neq u$ and $x$ is adjacent to 
$\lambda^{-1}(v)$ in $H$.
\end{itemize}
 Let 
$q_i''=|Z'''\cap Z_i|$ for $i\in \{1,\ldots,s\}$. By the definition of $\alpha_h'$, we obtain that 
\begin{equation}\label{eq:alphah-prime}
\alpha_{h}'(q_1'',\ldots,q_s'')\leq \sum_{z\in Z'''} w(z,\lambda(z)).
\end{equation}

Using (\ref{eq:alphah-1}) and (\ref{eq:alphah-prime}), we conclude that
\begin{align*}
\alpha_h(q_1,\ldots,q_s)=&\sum_{z\in Z'} w(z,\xi(z))=\big(\sum_{z\in Z''} w(z,\xi'(z))\big)+\big(\sum_{z\in Z'''} w(z,\lambda(z))\big)\\
\geq&\alpha_{h-1}(q_1',\ldots,q_s')+\alpha_{h}'(q_1'',\ldots,q_s''),
\end{align*}
and this immediately implies (\ref{eq:lower}).

Now we prove that 
\begin{equation}\label{eq:upper}
\alpha_{h}(q_1,\ldots,q_s)\leq \min \{\alpha_{h-1}(q_1',\ldots,q_s')+\alpha_{h}'(q_1'',\ldots,q_s'')\}.
\end{equation}
Consider  $s$-tuples $(q_1',\ldots,q_s')$ and $(q_1'',\ldots,q_s'')$ such that $q_i=q_i'+q_i''$ for $i\in\{1,\ldots,s\}$ for which the minimum in the right part of (\ref{eq:upper}) is achieved.
If $\alpha_{h-1}(q_1',\ldots,q_s')=+\infty$ or $\alpha_{h}'(q_1'',\ldots,q_s'')=+\infty$, then (\ref{eq:upper}) is trivial. Assume that 
$\alpha_{h-1}(q_1',\ldots,q_s')<+\infty$ and $\alpha_{h}'(q_1'',\ldots,q_s'')<+\infty$.

Since $\alpha_{h-1}(q_1',\ldots,q_s')<+\infty$, there is an injective mappings $\xi'\colon Z''\rightarrow (V(F_0)\setminus \psi(X\cup Y))\cup V(F_1)\cup\ldots\cup V(F_{h-1})$ satisfying (a$'$) and (b$'$) such that $\alpha_{h-1}(q_1',\ldots,q_s')=\sum_{z\in Z''} w(z,\xi(z))$, where $Z''\subseteq Z$ with $|Z''\cap Z_i|=q_i'$ for $i\in\{1,\ldots,s\}$. 
Because $\alpha_{h}'(q_1'',\ldots,q_s'')<+\infty$, there is an injection $\lambda$  of $Z'''$ to $V(F_h)$ such that (a$^{**}$) and (b$^{**}$) are fulfilled and $\alpha_{h}'(q_1'',\ldots,q_s'')=\sum_{z\in Z'''} w(z,\lambda(z))$ for $Z'''\subseteq Z$ with $|Z'''\cap Z_i|=q_i'$ for $i\in\{1,\ldots,s\}$. 

Recall that the values of 
$\alpha_{h-1}(q_1',\ldots,q_s')$ and $\alpha_h'(q_1'',\ldots,q_s'')$  depend only on the values of $q_1',\ldots,q_s'$ and $q_1'',\ldots,q_s''$ respectively, because if $x,y\in Z_i$, then $w(x,v)=w(y,v)$ for $v\in V(G)\setminus \psi(X\cup Y)$. Hence, we can assume that $(Z'',Z''')$ is a partition of $Z'$.
We define $\xi\colon  Z'\rightarrow (V(F_0)\setminus \psi(X\cup Y))\cup V(F_1)\cup\ldots\cup V(F_{h})$ by setting
$$
\xi(z)=
\begin{cases}
\xi'(z), &\mbox{if~} z\in Z'',\\
\lambda(z), &\mbox{if~} z\in Z'''.
\end{cases}
$$
Because $\xi'$ and $\lambda$ are injections and $\xi'(Z'')\cap \lambda(Z''')=\emptyset$, $\xi$ is an injection. Since $\xi'$ and $\lambda$ satisfy (a$'$), (b$'$) and (a$^{**}$), (b$^{**}$) respectively, $\xi$ satisfies (a) and (b). Therefore,
 \begin{align*}
\alpha_h(q_1,\ldots,q_s)\leq&\sum_{z\in Z'} w(z,\xi(z))=\big(\sum_{z\in Z''} w(z,\xi'(z))\big)+\big(\sum_{z\in Z'''} w(z,\lambda(z)))\\
=&\alpha_{h-1}(q_1',\ldots,q_s')+\alpha_{h}'(q_1'',\ldots,q_s''),
\end{align*}
and (\ref{eq:upper}) follows.

Combining (\ref{eq:lower}) and (\ref{eq:upper}), we obtain that (\ref{eq:final-dp}) holds, and this completes the correction proof of our algorithm.

To evaluate the running time, observe that there are at most $|V(G)|^s$ $s$-tuples $(q_1,\ldots,q_s)$. Since $s\leq 2^t$, it implies that the table of values of $\alpha_0(q_1,\ldots,q_s)$ can be computed  in time $|V(G)|^{\Oh(2^t)}\cdot \log W$. Similarly, the table of values of $\alpha_h'(q_1,\ldots,q_s)$ for each $h\in\{1,\ldots,r\}$ can be computed in the same time. To compute 
$\alpha_h(q_1,\ldots,q_s)$ for a given $s$-tuple $(q_1,\ldots,q_s)$ using (\ref{eq:final-dp}), we have to consider at most $|V(G)|^s$ pairs of $s$-tuples   $(q_1',\ldots,q_s')$ and
 $(q_1'',\ldots,q_s'')$. Hence, we can compute the table of values $\alpha_{h}(q_1,\ldots,q_s)$ from the tables of values of $\alpha_{h-1}(q_1,\ldots,q_s)$ and $\alpha_{h}'(q_1,\ldots,q_s)$ in time $|V(G)|^{\Oh(2^t)}\cdot \log W$ for each $h\in\{1,\ldots,r\}$. We conclude that the total running time is $|V(G)|^{\Oh(2^t)}\cdot\log W$.
\end{proof}

\subsection{Hardness of structured augmentation}\label{sec:hard}
In this section we show that Theorems~\ref{thm:wsca-one} and \ref{thm:wsca-two} are tight in the sense that if the vertex-cover number of graphs in a hereditary graph class $\mathcal{C}$ is unbounded, then
both structured augmentation problems are \classNP-complete.  Our hardness proof actually holds for for any $k$-edge connectivity augmentation.  
For a positive integer $k$, we define the following problem:

\defproblema{\probWSCA}%
{Graphs $G$ and $H$ such that $G$ is edge $(k-1)$-connected, a weight function $\cost\colon \binom{V(G)}{2}\rightarrow \mathbb{N}_0 $  and a nonnegative integer $\weight$.}%
{Decide whether there is an injective  $\varphi\colon V(H)\rightarrow V(G)$ such that $F=G\oplus_{\varphi}H$ is edge $k$-connected and the \emph{weight} of the mapping $\cost(\varphi)=\sum_{xy\in E(H)}\cost(\varphi(x)\varphi(y))\leq \weight$.}

Let us note that for $k=1$ this is \probStrucAugm 
and for $k=2$ this is \probStrucAugmtwo.
%
  Also we observe that it is unlikely that we can avoid the dependency on $t$ in the exponents of polynomial bounding the running time when solving   \probWSCA for $H$ with $\beta(H)\leq t$.

Recall that the \textsc{Subgraph Isomorphism} problem asks, given two graphs $G$ and $H$, whether $G$ contains $H$ as a (not necessarily induced) subgraph. We can observe that \probWSCA  when $H$ restricted to be 
in a graph class $\mathcal{C}$ is at least as hard as \textsc{Subgraph Isomorphism} with the same restriction.

\begin{lemma}\label{lem:subgr-isom}
Let $\mathcal{C}$ be a graph class.
If  \textsc{Subgraph Isomorphism} is \classNP-complete for $H\in \mathcal{C}$, then for every positive integer $k$, \probWSCA is \classNP-complete for $H\in \mathcal{C}$ even if the weight of every pair of vertices of $G$ is restricted to be ether $0$ or $1$. Also if   \textsc{Subgraph Isomorphism}  is $\classW{1}$-hard for $H\in \mathcal{C}$ when parameterized by $|V(H)|$, then so is  \probWSCA.
\end{lemma}

\begin{proof}
Let $k$ be a positive integer, and let $(G,H)$ be an instance of \textsc{Subgraph Isomorphism}. Assume without loss of generality that  $|V(G)|>k$. We construct the complete graph $F$ with the set of vertices $V(G)$ and define the weight function $\cost\colon \binom{V(G)}{2}\rightarrow\{0,1\}$ 
by setting 
\[\cost(uv)=
\begin{cases}
0, &\mbox{if~} uv\in E(G),\\
1, &\mbox{if~} uv\notin E(G).
\end{cases}
\]
Then we let $\weight=0$. Notice that $F$ is $k$-connected and $H$ is a subgraph of $G$ if and only if there is an injection $\varphi\colon V(H)\rightarrow V(G)$ with $\cost(\varphi)=\sum_{xy\in E(G) }\cost(\varphi(x)\varphi(y))=0$. Then $(G,H)$ is a yes-instance of \textsc{Subgraph Isomorphism} if and only if $(G,H,\cost,\weight)$ is a yes-instance of \probWSCA and the claim follows.
\end{proof}

The \textsc{Clique} problem asks, given a graph $G$ and a positive integer $k$, whether $G$ has a clique of size $k$ or, in other words, whether the complete graph $K_k$ is a subgraph of $G$. 
It is well-known that \textsc{Clique} is \classNP-complete~\cite{GareyJ79}.
Then Lemma~\ref{lem:subgr-isom} implies the following lemma.

\begin{lemma}\label{lem:clique}
Let $\mathcal{C}$ be a hereditary graph class that contains $K_n$ for arbitrary positive integer $n$.
Then for every positive integer $k$, \probWSCA is \classNP-complete for $H\in \mathcal{C}$ even if the weight of every pair of vertices of $G$ is restricted to be ether $0$ or $1$.
\end{lemma}

 Let us note that \textsc{Clique} is  \classW{1}-hard when parameterized by $k$, see the book of Downey and Fellows~\cite{DowneyF13} for an introduction to parameterized complexity.
Notice that $\beta(K_k)=k-1$. 
Then Lemma~\ref{lem:subgr-isom} implies the following proposition.

\begin{proposition}\label{prop:vc-w-hard}
For every positive integer $k$, \probWSCA is \classW{1}-hard when parameterized by $\beta(H)$ even if the weight of every pair of vertices of $G$ is restricted to be ether $0$ or $1$.
\end{proposition}

This proposition implies that unless \classFPT$=$\classW{1}, we cannot solve  \probWSCA for $k=1,2$ in time $f(\beta(H))\cdot |V(G)|^{\Oh(1)}$.  Hence the running time of the form 
$|V(G)|^{f(t)}$
 of algorithms solving   \probWSCA for graphs $H$ with $\beta(H)\leq t$ is probably unavoidable. 

%

The \textsc{Balanced Biclique} asks, given a graph $G$ and a positive integer $k$, whether $G$ contains $K_{k,k}$ as a subgraphs. It is known that \textsc{Balanced Biclique} is \classNP-complete~\cite{GareyJ79}. Using  Lemma~\ref{lem:subgr-isom} we obtain the next lemma.

 \begin{lemma}\label{lem:biclique}
Let $\mathcal{C}$ be a hereditary graph class that contains $K_{n,n}$ for arbitrary positive integer $n$.
Then for every positive integer $k$, \probWSCA is \classNP-complete for $H\in \mathcal{C}$ even if the weight of every pair of vertices of $G$ is restricted to be ether $0$ or $1$.
\end{lemma}

Now we consider \probWSCA for $k\geq1$ for matching graphs. 

\begin{lemma}\label{lem:match-hard-one}
Let $\mathcal{C}$ be a hereditary graph class that contains  a matching graph of arbitrary size.
Then \probStrucAugm is \classNP-complete for $H\in \mathcal{C}$ even if the weight of every pair of vertices of $G$ is at most $2$.
\end{lemma}

\begin{proof}
Clearly, it is sufficient to prove that \probStrucAugm is \classNP-complete if $H$ is a matching graph.
We reduce from the \textsc{Hamiltonian Path} problem. Recall that this problem asks whether a graph $G$ has a path containing all the vertices of $G$.  \textsc{Hamiltonian Path} is known to be \classNP-complete for cubic graphs~\cite{GareyJ79}.

Let $G$ be a cubic graph with $n$ vertices. We construct the graph $G'$ as follows.
\begin{itemize}
\item Construct a copy of $V(G)$.
\item For each edge $e=uv\in E(G)$, construct two vertices $u^e$ and $v^e$ and make them adjacent to $u$ and $v$ respectively.
\end{itemize}
Notice that $G'$ is the disjoint union of $n$ copies of $K_{1,3}$. 
We define $H$ to be the matching graph with $2n-1$ edges. 
Now we define $\cost\colon \binom{V(G')}{2}\rightarrow \mathbb{N}_0$. For each edge $e\in E(G)$, we set
$\cost(uu^e)=\cost(vv^e)=0$ and $\cost(u^ev^e)=1$. For all remaining pairs of distinct vertices $x$ and $y$, we set $\cost(xy)=2$. Finally, let $\weight=n-1$.

We claim that $G$ has a Hamiltonian path if and only if $(G',H,\cost,\weight)$ is a yes-instance of \probStrucAugm.

Suppose that $P=v_1\ldots v_n$ is a Hamiltonian path in $G$. Denote by $x_1y_1,\ldots,x_{2n-1}y_{2n-1}$ the edges of $H$. We consider the following injection $\varphi\colon V(H)\rightarrow V(G')$:
\begin{itemize}
\item for $i\in\{1,\ldots,n-1\}$, set $\varphi(x_i)=v_i^{v_iv_{i-1}}$ and $\varphi(y_i)=v_{i+1}^{v_iv_{i+1}}$,
\item for each $i\in\{1,\ldots,n\}$, find an edge $e$ in $G$ incident to $v_i$ such that $e\notin E(P)$ and then set $\varphi(x_{n-1+i})=v_i$ and $\varphi(y_{n-i+1})=v_{i+1}^e$.
\end{itemize}
It is straightforward to verify that $F=G'\oplus_\varphi H$ is connected and $\cost(\varphi)=\sum_{xy\in E(H)}\cost(\varphi(x)\varphi(y))=n-1\leq \weight$.
 
Assume now that there is an injection $\varphi\colon V(H)\rightarrow V(G')$ such that $F=G'\oplus_\varphi H$ is connected and $\cost(\varphi)=\sum_{xy\in E(H)}\cost(\varphi(x)\varphi(y))\leq n-1= \weight$. Let $A=E(F)\setminus E(G')$. Observe that $|A|\geq n-1$, because  $G'$ contains $n$ components. For each $a\in A$, $\cost(a)\geq 1$ and $\cost(a)=1$ if and only if $a=u^ev^e$ for some edge $e=uv\in E(G)$. Then $A$ contains exactly $n-1$ edges and for each $a\in A$, there is $e_a=uv\in E(G)$ such that $a=u^{e_a}v^{e_a}$. Notice that the edges $e_a$ for $a\in A$ are pairwise distinct, because $\varphi$ is an injection. 
Let $X=\{xy\in E(H)\mid \varphi(x)\varphi(y)\in A\}$ and $Y=E(H)\setminus X$. 
Since $H$ has $2n-1$ edges and $|A|=n-1$, $|Y|=n$. Because $\cost(A)= \weight$, we have that for each $xy\in Y$, $\cost(\varphi(x)\varphi(y))=0$, that is, $\varphi(x)\varphi(y)\in E(G')$.
Because each component of $G'$ is a copy of $K_{1,3}$, we have that for each vertex $u\in V(G)$, there is an edge $e$ incident to $u$ in $G$ such that for an edge $xy\in Y$, $\varphi(\{x,y\})=\{u,u^e\}$. Because $\varphi$ is an injective mapping, this implies that at most two edges of $A$ have their endpoints in the same component of $G'$. Therefore,
every vertex $v\in G$ is incident to at most two edges of the set  $B=\{e_a\mid a\in A\}$. Since $F$ is connected, we conclude that the edges of $B$ compose a Hamiltonian path in $G$. 
\end{proof}

\begin{lemma}\label{lem:match-hard-two}
Let $\mathcal{C}$ be a hereditary graph class that contains  a matching graph of arbitrary size.
Then for every $k\geq 2$, \probWSCA is \classNP-complete for $H\in \mathcal{C}$ in the strong sense.
\end{lemma}

\begin{proof}
Let $k\geq 2$ be an integer.
Clearly, it is sufficient to prove that \probWSCA is \classNP-complete if $H$ is a matching graph.
We reduce from the \textsc{Biconnectivity Augmentation} problem that asks, given a graph $G$, a weight function $c\colon\binom{V(G)}{2}\setminus E(G)\rightarrow \mathbb{N}_0$ and a positive integer $w$, whether there is $A\subseteq \binom{V(G)}{2}\setminus E(G)$ with $c(A)\leq w$ such that the graph $G'$ obtained from $G$ by the addition of edges of $A$ is 2-connected. 
This problem was shown to be \classNP-complete by Frederickson and J{\'{a}}J{\'{a}} in~\cite{FredericksonJ81} even if  $G$ restricted to be a tree and $c(uv)\in\{1,2\}$ for $uv\in \binom{V(G)}{2}\setminus E(G)$.

Let $(T,c, \weight)$ be an instance of \textsc{Biconnectivity Augmentation} where $T$ is a tree and  $c\colon\binom{V(T)}{2}\setminus E(T)\rightarrow \{1,2\}$. 
Let $r=\max\{k, \weight\}$.
We construct the graph $G$ as follows.
\begin{itemize}
\item Construct a copy of $V(T)$.
\item For each $e\in E(T)$, construct a clique $Q_e$ of size $k$ and make the vertices of $Q$ adjacent to $u$ and $v$.
\item For each $u\in V(T)$, construct a clique $R_u$ of size $2r$, denote its vertices by $x_1^u,\ldots,x_r^u,y_1^u,\ldots,y_r^u$ and make them adjacent to $u$.
\end{itemize}
Observe that $G$ is $(k-1)$-connected and its minimum edge separators correspond to the edges of $T$. More precisely, for an edge $e=uv\in V(T)$, $G$ has two minimum separators $S_1=\{uz\in E(G)\mid z\in Q_e\}$ and $S_2=\{vz\in E(G)\mid z\in Q_e\}$.
We define $H$ to be the matching graph with $ \weight$ edges and denote its edges by $p_1q_1,\ldots,p_ \weight q_ \weight$.
Finally, we define the weight function $\binom{V(G)}{2}\rightarrow\mathbb{N}_0$ as follows.
\begin{itemize}
\item For each $uv\in \binom{V(T)}{2}\setminus E(T)$, $\cost(x_i^uy_j^v)=c(uv)$ for $i,j\in\{1,\ldots,r\}$.
\item For each $u\in V(T)$, $\cost(x_i^uy_i)=0$.
\item We set $\cost(pq)=C+1$ for the remaining pairs of distinct vertices of $G$.
\end{itemize}

We claim that $(T,c, \weight)$ is a yes-instance of \textsc{Biconnectivity Augmentation} if and only if $(G,H,\cost,\weight)$ is a yes-instance of \probWSCA.

Suppose that $(T,c, \weight)$ is a yes-instance of \textsc{Biconnectivity Augmentation}. Then there is  $A\subseteq \binom{V(G)}{2}\setminus E(G)$ with $c(A)\leq w$ such that the graph $T'$ obtained from $T$ by the addition of edges of $A$ is 2-connected.  Let $A=\{a_1b_1,\ldots a_sb_s\}$. Since $c(e)\in \{1,2\}$ for $e\in E(G)$, $s\leq  \weight$. We construct the injective mapping $\varphi\colon V(H)\rightarrow V(G)$ as follows.
\begin{itemize}
\item For $i\in \{1,\ldots,s\}$, if  $a_ib_i=uv$ for nonadjacent distinct $u,v\in V(T)$, then set $\varphi(p_i)=x_u^i$ and $\varphi(q_i)=x_v^i$.
\item For $i\in \{s+1, \weight\}$, select $u\in V(T)$ and set $\varphi(p_i)=x_i^u$ and $\varphi(q_i)=y_i^u$. 
\end{itemize}
By the definition of $\cost$, $\cost(\varphi)=\sum_{i=1}^w \cost(\varphi(p_i)\varphi(q_i))=c(A)\leq  \weight$. Consider $F=G\oplus_\varphi H$. Recall that $T'$ is 2-connected. Then for every bridge $uv$ of $T$, that is, for every edge $uv$, there is $a_ib_i$ that covers $uv$ by Observation~\ref{obs:cover-bridge}. By the definition of $\varphi$, there is an edge $e\in E(F)\setminus E(G)$ with its endpoints in $R_{a_i}$ and $R_{b_i}$. Then this edges covers the separators $S_1=\{uz\in E(G)\mid z\in Q_{uv}\}$ and $S_2=\{vz\in E(G)\mid z\in Q_{uv}\}$. It implies that all edge separators of $G$ of size $k-1$ are covered by edges of $E(F)\setminus E(G)$. By Observation~\ref{obs:cover-bridge}, we conclude that $F$ is $k$-connected. Therefore, $(G,H,\cost,\weight)$ is a yes-instance of \probWSCA.

Suppose now that $(G,H,\cost,\weight)$ is a yes-instance of \probWSCA. Then there is an injection $\varphi\colon V(H)\rightarrow V(G)$ such that $\cost(\varphi)\leq  \weight$ and $F=G\oplus_\varphi H$ is $k$-connected. By Observation~\ref{obs:cover-bridge}, we have that for each edge $uv\in E(T)$, there is $e\in E(F)\setminus E(G)$ such that $e$ covers the separator 
$\{uz\in E(G)\mid z\in Q_{uv}\}$ of $G$. Since $\cost(e)\leq  \weight$, we obtain that $e$ has its end vertices in $x_{u'}^i$ and $x_{v'}^j$ for some $i,j\in\{1,\ldots,r\}$ and two nonadjacent $u',v'\in V(T)$. Denote $a_e=u'v'$.
Notice that if we add $a_e$ to $T$, then $a_e$ covers $uv$ in the obtained graph. Observe also that $c(a_e)=\cost(e)$. We consider the set $A$ of distinct $a_e\in \binom{V(T)}{2}\setminus E(T)$ constructed in the described way  for $e\in E(F)\setminus E(G)$ covering the separators $\{uz\in E(G)\mid z\in Q_{uv}\}$ of $G$. We have that the graph $T'$ obtained from $T$ by the addition of the edges of $A$ is 2-connected by Observation~\ref{obs:cover-bridge}.
Since $c(A)\leq \cost(E(F)\setminus E(G))\leq\cost(\varphi)\leq  \weight$, we conclude that $(T,c, \weight)$ is a yes-instance of \textsc{Biconnectivity Augmentation}.
\end{proof}

To classify the computational complexity of \probWSCA for hereditary graph classes, we use  the Ramsey's theorem (see, e.g.,~\cite{Diestel12} for the introduction).
For two positive integers $p$ and $q$, we denote by $R(p,q)$ the \emph{Ramsey number}, that is, the smallest $n$ such that every graph on $n$ vertices has either a clique of
size $p$ or an independent set of size $q$. Those numbers are all finite by the Ramsey's theorem. In particular, Marx and Wollan in~\cite{MarxW15} observed the following corollary.

\begin{lemma}\label{lem:ramsey}
Let $H$ be a graph and $n\geq 1$ a positive integer. If $H$ contains a matching with $5^{10n}$ edges, then $H$ either contains the matching graph with $n$ edges as an induced subgraph or $H$ contains $K_{n,n}$ as a subgraph. 
\end{lemma} 

Now we are ready to prove the main theorem of this section.

\begin{theorem}\label{thm:main-hard}
Let $k$ be a positive integer. Let also $\mathcal{C}$ be a hereditary graph class. Then if the vertex-cover number of   $\mathcal{C}$ is unbounded, then \probWSCA is \classNP-complete for $H\in \mathcal{C}$ in the strong sense.
\end{theorem}

\begin{proof}
Suppose that the vertex-cover number of   $\mathcal{C}$ is unbounded. It means that for any positive integer $n$, $\mathcal{C}$ contains a graph with a matching $M$ of size at least $n$.
Lemma~\ref{lem:ramsey} implies that  either  $\mathcal{C}$ contains matching graphs of arbitrary size, or for any positive $n$, $\mathcal{C}$ contains a graph $H$ that has a spanning subgraph isomorphic to $K_{n,n}$. If $\mathcal{C}$ contains matching graphs of arbitrary size, then \probWSCA is \classNP-complete for $H\in \mathcal{C}$ by Lemmata~\ref{lem:match-hard-one} and 
\ref{lem:match-hard-two} for $k=1$ and $k\geq 2$ respectively.
Suppose that for any positive $n$, $\mathcal{C}$ contains a graph $H$ that has a spanning subgraph isomorphic to $K_{n,n}$. If $K_r\in \mathcal{C}$ for $r\geq 1$, then \probWSCA is \classNP-complete for $H\in \mathcal{C}$ by Lemma~\ref{lem:clique}. Assume that there is a constant $p\geq 1$ such that $K_r\notin \mathcal{C}$ for $r\geq p$. Then for any positive $q$,  we have that
$\mathcal{C}$ contains a graph $H$ that has a spanning subgraph isomorphic to $K_{n,n}$ for $n=R(p,q)$. It immediately implies that $K_{q,q}$ is an induced subgraph of $H$. Therefore, for every positive $n$, $K_{n,n}\in \mathcal{C}$, and  \probWSCA is \classNP-complete for $H\in \mathcal{C}$ by Lemma~\ref{lem:biclique}. 
\end{proof}

\section{Augmenting unweighted graphs}\label{sec:unweighted}
In this section we investigate unweighted \probStrucAugm and \probStrucAugmtwo.
Let us remind that in the unweighted cases of the structured augmentation problems the task is to identify whether there is a superposition of graphs $G$ and $H$ of  edge connectivity $1$ or $2$, correspondingly.   
In other words,  we have the weight  $\cost(uv)=0$ for every pair of vertices of $G$ and $W=0$. 
  We obtain structural characterizations of yes-instances for both problems.  
 
\subsection{Unweighted \probStrucAugm}
It is convenient to consider the special case when $H$ is connected separately.

\begin{lemma}\label{lem:conn-H}
Let $G$ and $H$ be graphs such that $|V(H)|\leq |V(G)|$ and $H$ is connected.
Then there is an injective mapping $\varphi\colon V(H)\rightarrow V(G)$ such that $F=G\oplus_\varphi H$ is connected if and only if $c(G)\leq |V(H)|$.
\end{lemma}

\begin{proof}
Suppose that there is an injective mapping $\varphi\colon V(H)\rightarrow V(G)$ such that $F=G\oplus_\varphi H$ is connected. Then for each component $G'$ of $G$, there is $v\in V(G')$ such that $v\in \varphi(V(H))$. Since $\varphi$ is injective, then $c(G)\leq |V(H)|$.

Assume now that $c(G)\leq |V(H)|$.  Let $G_1,\ldots,G_s$ be the components of $G$. Since $|V(H)|\leq |V(G)|$, there are distinct vertices $x_1,\ldots,x_s\in V(H)$.
We select arbitrarily a vertex $v_i\in V(G_i)$ for $i\in\{i,\ldots,s\}$.
We construct the injective mapping $\varphi\colon V(H)\rightarrow V(G)$ as follows. We set $\varphi(x_i)=v_i$ for $i\in\{1,\ldots,s\}$ and then extend $\varphi$ on other vertices of $H$ selecting their images in $V(G)\setminus\{v_1,\ldots,v_s\}$ arbitrarily. It is straightforward to verify that  $F=G\oplus_\varphi H$ is connected.
\end{proof}

Now we consider the case when $H$ is not connected.

\begin{lemma}\label{lem:disconn-H}
Let $G$ and $H$ be graphs such that $|V(H)|\leq |V(G)|$, $H$ has no isolated vertices and is disconnected.
Then there is an injective mapping $\varphi\colon V(H)\rightarrow V(G)$ such that $F=G\oplus_\varphi H$ is connected if and only if 
\begin{itemize}
\item[(i)] $i(G)\leq |V(H)|-c(H)$, and
\item[(ii)] $c(G)\leq |V(H)|-c(H)+1$.
\end{itemize}
\end{lemma}

\begin{proof}
Suppose that there is an injective mapping $\varphi\colon V(H)\rightarrow V(G)$ such that $F=G\oplus_\varphi H$ is connected. 

We prove  (i) by contradiction. Assume that $i(G)>|V(H)|-c(H)$. Then there is a component $H'$ of $H$ such that every vertex of $\varphi(V(H'))$ is an isolated vertex of $G$. We obtain that  $F[\varphi(V(H'))]$ isomorphic to $H'$ is a component of $F$ contradicting the connectivity of $F$.

To show (ii), denote by $G_1,\ldots,G_s$ and $H_1,\ldots,H_r$ the components of $G$ and $H$ respectively. Consider the auxiliary bipartite graph $R$ whose vertices are the components of $G$ and $H$,  $\{G_1,\ldots,G_s\}$ and $\{H_1,\ldots,H_r\}$ form the bipartition of the vertex set, and $G_i$ is adjacent to $H_j$ if and only if $\varphi$ maps a vertex of $H_j$ to a vertex of $G_i$.
Since $F$ is connected, we obtain that $R$ is connected as well. Therefore, $|V(R)|\leq |E(R)|+1$. Because $\varphi$ is an injection, $|E(R)|\leq |V(H)|$. Then
$$c(G)+c(H)=s+r=|V(R)|\leq |E(R)|+1\leq |V(H)|+1$$
and (ii) follows.

Suppose now that (i) and (ii) are fulfilled. Denote by $H_1,\ldots,H_r$ the components of $H$. 

Assume that $i(G)=|V(H)|-c(H)$. By (ii), we have that $G$ has at most $i(G)+1$ components.  Hence, because $|V(H)|\leq|V(G)|$, $G$ has exactly $i(G)+1$ components: $i(G)$ isolated vertices and a component $G'$ with at least $c(H)$ vertices. We select a vertex $x_i$ in each component $H_i$ for $i\in \{1,\ldots,s\}$ and $s$ distinct vertices $v_1,\ldots,v_s$ in $G'$. We  
construct the injective mapping $\varphi\colon V(H)\rightarrow V(G)$ as follows. We set $\varphi(x_i)=v_i$ for $i\in\{1,\ldots,s\}$ and then extend $\varphi$ on other vertices of $H$ by mapping them into isolated vertices of $G$.  It is straightforward to verify that  $F=G\oplus_\varphi H$ is connected.

Suppose from now that $i(G)<|V(H)|-c(H)$. We select the minimum $h\in\{1,\ldots,r\}$ such that $(\sum_{j=1}^h|V(H_j)|)-h>i(G)$. For each $i\in\{1,\ldots,h-1\}$, we select $x_i\in V(H_i)$ if $h>1$.
We start constructing the injective mapping $\varphi\colon V(H)\rightarrow V(G)$ by mapping the vertices of $V(H_i)\setminus\{x_i\}$ to isolated vertices of $G$. Then we map $|V(H_h)|-\sum_{i=1}^{h-1}(|V(H_i)|-1)$ vertices of $H_h$ to the remaining isolated vertices of $G$. Notice that by the choice of $h$, at least 2 vertices of $H_h$ and the vertices of $H_{h_1},\ldots,H_r$ are not mapped yet. Denote by $W$ the set of these vertices.
 Recall also that each component of $H$ has at least 2 vertices. Denote by $G_1,\ldots,G_s$ the components of $G$ with at least 2 vertices each. Since   $|V(H)|\leq |V(G)|$, we have that  these components exist and that the total number of vertices in these components is at least $W+(h-1)$. By (ii), we have that 
$s+i(G)\leq |V(H)|-c(H)+1$. Therefore, $s\leq |W|-(r-h+1)+1$.
   
If $s\leq r-h+2$, then we select  $v_i\in V(G_i)$ for $i\in\{1,\ldots,s-1\}$ and $v_i'\in V(G_i)$ for $i\in\{2,\ldots,s\}$ such that $v_i\neq v_i'$. Then for each $i\in\{1,\ldots,s-1\}$, we pick two vertices in $H_{h+i-1}$ and map them to $v_i$ and $v_{i+1}'$ respectively. The remaining vertices of $W$ and the vertices $x_1,\ldots,x_{h-1}$ are mapped into distinct vertices of $G$ that were not used for constructing $\varphi$ yet. It is again straightforward to see that $F=G\oplus_\varphi H$ is connected.

Suppose that $s> r-h+2$. We select  $v_i\in V(G_i)$ for $i\in\{1,\ldots,r-h+1\}$ and $v_i'\in V(G_i)$ for $i\in\{2,\ldots,r-h+2\}$ such that $v_i\neq v_i'$. Then for each $i\in\{h,\ldots,r\}$, we pick two vertices in $H_i$ and map them to $v_{i-h+1} $ and $v_{i-h+2}'$ respectively. For every $i\in \{r-h+3,\ldots,s\}$, we pick a vertex $u_i\in G_i$ and $y_i\in W$ that is not mapped yet. Notice that since $s\leq |W|-(r-h+1)+1$, this selection is possible. Then we set $\phi(y_i)=u_i$ for $i\in \{r-h+3,\ldots,s\}$. The remaining vertices of $W$ and the vertices $x_1,\ldots,x_{h-1}$ are mapped into distinct vertices of $G$ that were not used for constructing $\varphi$ yet. Again, we have that $F=G\oplus_\varphi H$ is connected. 
\end{proof}

Lemmata~\ref{lem:conn-H} and \ref{lem:disconn-H} 
immediately imply the following theorem.

\begin{theorem}\label{thm:SCA-one}
Let $G$ and $H$ be graphs such that $H$ has no isolated vertices and $|V(H)|\leq|V(G)|$. 
Then there is an injective mapping $\varphi\colon V(H)\rightarrow V(G)$ such that $F=G\oplus_\varphi H$ is connected if and only if 
  $c(G)\leq |V(H)|-c(H)+1$ and one of the following holds:
\begin{itemize}
\item[(i)] $H$ is connected,
\item[(ii)] $H$ is disconnected graph and $i(G)\leq |V(H)|-c(H)$.
\end{itemize}
\end{theorem}

The next statement is a straightforward corollary of Theorem~\ref{thm:SCA-one}.

\begin{corollary}\label{cor:SCA-one}
Unweighted \probStrucAugm is  solvable in time $\Oh(|V(G)|+|E(G)|+|E(H)|)$.
\end{corollary} 

\subsection{Unweighted \probStrucAugmtwo}\label{sec:unweightwoconnect}
Now we consider the case  \probStrucAugmtwo. 
Our structural results are based on the following observation.

\begin{observation}\label{obs:pend}
Let $G$ and $H$ be graphs and let $\varphi\colon V(H)\rightarrow V(G)$ be an injective mapping such that $F=G\oplus_\varphi H$ is $2$-connected.
Then for every pendant biconnected component $G'$ of $G$, there is $x\in V(H)$ such that $\varphi(x)\in V(G')$.
\end{observation}

In particular, Observation~\ref{obs:pend} implies the following.

\begin{observation}\label{obs:pend-cover}
Let $G$ and $H$ be graphs 
and let $\varphi\colon V(H)\rightarrow V(G)$ be an injective mapping with the property that $F=G\oplus_\varphi H$ is $2$-connected.
Then $p(G)\leq |V(H)|$.
\end{observation}

To simplify the proofs of our structural lemmata, we use the following straightforward observation.

\begin{observation}\label{obs:subgr}
Let $G$ and $H$ be graphs such that $|V(H)|\leq |V(G)|$. If $H$ has a subgraph $H'$ such that 
 there is an injective mapping $\varphi\colon V(H')\rightarrow V(G)$ with the property that $F'=G\oplus_\varphi H'$ is $k$-connected, then for every injective extension $\psi$ of $\varphi$ on $V(H)$, $F=G\oplus_\varphi H$ is $k$-connected.
\end{observation}

This allows us to use the following strategy to increase the connectivity of a graph $G$. Let $\ell=p(G)\geq 2$. We select $\ell$ pairwise nonadjacent vertices $v_1,\ldots,v_\ell$ in distinct pendant biconnected components of $G$. Then we find an induced subgraph $H'$ of $H$ with $\ell$ vertices and construct a bijection  $\varphi\colon V(H')\rightarrow \{v_1,\ldots,v_\ell\}$ with the property that $F=G\oplus_\varphi H'$ is 2-connected. Notice that if $H_1,\ldots,H_r$ are the components of $H'$, then $F=G\oplus_\varphi H'=(\ldots((G\oplus_{\varphi_1} H_1)\oplus_{\varphi_2}H_2)\ldots \oplus_{\varphi_r}H_r$ where $\varphi_i=\varphi|_{V(H_i)}$ for $i\in\{1,\ldots,r\}$. The construction of $\varphi$ is inductive and is based on the following lemma.

\begin{lemma}\label{lem:components}
Let $G$ be a connected graph with $\ell=p(G)\geq 2$. Let also $\mathcal{P}=\{P_1,\ldots,P_\ell\}$ be the set of pendant biconnected components of $G$, $v_i\in V(P_i)$ for $i\in\{1,\ldots,\ell\}$ and $v_1,\ldots,v_\ell$ are pairwise nonadjacent. Let also $H$ be a connected graph such that $2\leq|V(H)|\leq\ell$ and $\ell-|V(H)|\neq 1$. 
Then there is an injective mapping $\varphi\colon V(H)\rightarrow\{v_1,\ldots,v_\ell\}$ such that for $F=G\oplus_\varphi H$, the set of  pendant biconnected components is $\mathcal{P}\setminus\{P_i\mid v_i\in \varphi(V(H))\}$.
\end{lemma}

\begin{proof}
Suppose first that $|V(H)|=\ell$. Because $H$ is connected and the vertices $v_1,\ldots,v_\ell$ are pairwise nonadjacent, we have that for any bijection $\varphi\colon V(H)\rightarrow\{v_1,\ldots,v_\ell\}$, $F=G\oplus_\varphi H$ is 2-connected by Observations~\ref{obs:cover-bridge} and \ref{obs:edge-add}. Then the set of pendant biconnected components of $F$ is empty and the claim of the lemma holds. 

Assume from now that $|V(H)|<\ell$. Since $\ell-|V(H)|\neq 1$, $|V(H)|\leq\ell-2$.

Consider the graph $T$ obtained by contracting edges of each biconnected component of $G$. To simplify notations, assume that the vertex obtained by contracting of each $P_i$ is $v_i$ for $i\in\{1,\ldots,\ell\}$. Notice that $v_1,\ldots,v_\ell$ are the leaves of $T$. Notice also that the edges of $T$ are exactly the bridges of $G$.
Also for every $e\in E(T)$, we have the following property: $e$ belongs to a $(v_i,v_j)$-path in $G$ if and only if $e$ belongs to the unique $(v_i,v_j)$-path in $T$.

Since $\ell\geq |V(H)|+2\geq 4$, there are two distinct leaves $v_s$ and $v_t$ of $T$ such that for the unique $(v_s,v_t)$-path $P$ in $T$, there are two leaves $v_i$ and $v_j$ that are in the distinct components of $T-V(P)$. We select $L\subseteq \{v_1,\ldots,v_\ell\}$ of size  $|V(H)|$ such that $v_s,v_t\in L$ and $v_i,v_j\notin L$. Let $\varphi\colon V(H)\rightarrow L$ be a bijection.
Because $H$ is connected, by Observation~\ref{obs:edge-add},  the vertices of $G$ from the biconnected components that are crossed by $(v_p,v_p)$-paths in $G$ for $v_p,v_q\in L$
 induce a biconnected component $Q$ of $F=G\oplus_\varphi H$. Moreover,  because of the choice of $v_i$ and $v_j$, at least two bridges of $G$ have incident vertices in $Q$. Therefore, $\mathcal{P}\setminus\{P_i\mid v_i\in \varphi(V(H))\}$ is the set of pendant biconnected components of $F$.
\end{proof}

Using Lemma~\ref{lem:components}, we obtain the next lemma.

\begin{lemma}\label{lem:comp-connect}
Let $G$ be a connected graph with $\ell=p(G)\geq 2$. Let also $\mathcal{P}=\{P_1,\ldots,P_\ell\}$ be the set of pendant biconnected components of $G$, $v_i\in V(P_i)$ for $i\in\{1,\ldots,\ell\}$ and $v_1,\ldots,v_\ell$ are pairwise nonadjacent. Let also $H$ be a  graph with $\ell$ vertices such that each component of $H$ contains at least 2 vertices. 
Then there is a bijection $\varphi\colon V(H)\rightarrow\{v_1,\ldots,v_\ell\}$ such that $F=G\oplus_\varphi H$ is $2$-connected.
\end{lemma}

\begin{proof}
The proof is by the induction on the number of components of $H$.
If $H$ is connected, them Lemma~\ref{lem:components} implies the claim. Assume that $H$ is disconnected and let $H'$ be a component of $H$. 

Because each component of $H$ has size at least 2 and $|V(H)|=\ell$,  $|V(H')|\geq 2$ and 
$\ell-|V(H)|\neq 1$. By Lemma~\ref{lem:components}, there is an injective mapping $\varphi'\colon V(H')\rightarrow\{v_1,\ldots,v_\ell\}$ such that for $G'=G\oplus_{\varphi'} H'$, the set of  pendant biconnected components is $\mathcal{P}'=\mathcal{P}\setminus\{P_i\mid v_i\in \varphi'(V(H'))\}$. Let $\ell'=|\mathcal{P}'|$.

Consider $H''=H-V(H')$. Clearly, $H''$ has less components than $H$. We also have that $|V(H'')|=\ell'$ and each component of $H''$ has size at least 2.
We apply the inductive hypothesis for $G'$ and $H''$. Hence there is a bijection
$\varphi''\colon V(H)\rightarrow\{v_1,\ldots,v_\ell\}\setminus\{v_i\mid v_i\in\varphi'(V(H'))\}$ such that $F=G'\oplus_{\varphi''} H''$ is $2$-connected.

For $x\in V(H)$, let  \[\varphi(x)=
\begin{cases}
\varphi'(x),&\mbox{if~} x\in V(H'),\\
\varphi''(x),& \mbox{if~} x\in V(H'').
\end{cases}
\]
Clearly, $\varphi$ maps $V(H)$ to $\{v_1,\ldots,v_\ell\}$ bijectively. 
Then $F=G'\oplus_{\varphi''} H''=(G\oplus_{\varphi'}H')\oplus_{\varphi''}H''=G\oplus_\varphi H$ and is 2-connected.
\end{proof}

Now we are ready to prove the main structural results for unweighted  \probSCAtwo.
First, we observe that the case when $G$ is 2-connected is trivial.

\begin{observation}\label{obs:two-conn}
Let $G$ and $H$ be graphs such that $|V(H)|\leq |V(G)|$ and $G$ is $2$-connected. Then for any injection  $\varphi\colon V(H)\rightarrow V(G)$, $F=G\oplus_\varphi H$ is $2$-connected. 
\end{observation}

From now we can assume that $G$ is connected but not 2-connected. In particular, $p(G)\geq 2$.
It is convenient to  consider separately the case when $H$ is a matching graph.

\begin{lemma}\label{lem:match}
Let $G$ be a connected graph and let $H$ be a matching graph with $2\leq p(G)\leq |V(H)|\leq |V(G)|$. Then there is an injection $\varphi\colon V(H)\rightarrow V(G)$ such that  $F=G\oplus_\varphi H$ is $2$-connected unless $G$ is a star $K_{1,n}$ where $n$ is odd.
\end{lemma}

\begin{proof}
If  $G$ is a star $K_{1,n}$ where $n$ is odd, then because $p(G)\leq |V(H)|\leq |V(G)|$, $|V(H)|=|V(G)|$. Then for every  injection $\varphi\colon V(H)\rightarrow V(G)$, there is an edge $xy\in E(H)$ such that $u=\varphi(x)$ is the central vertex of the star $G$ and $v=\varphi(y)$ is a leaf of $G$. We have that $uv$ is a bridge of $F=G\oplus_\varphi H$ and, therefore, $F$ is not $2$-connected.
Assume from now that $G$ is not a star with the odd number of leaves.

Let $\mathcal{P}=\{P_1,\ldots,P_\ell\}$ be the set of pendant biconnected components of $G$.  We select $v_i\in V(P_i)$ for $i\in\{1,\ldots,\ell\}$ in such a way that $v_1,\ldots,v_\ell$ are pairwise nonadjacent. Notice that it always can be done because $G$ is distinct from $K_2$ as this is the star $K_{1,1}$. 

Suppose that $\ell$ is even. Since $\ell\leq |V(H)|$, $H$ has an induced subgraph $H'$ that contains $\ell/2$ components with $\ell$ vertices. By Lemma~\ref{lem:comp-connect},    
there is a bijection $\varphi\colon V(H')\rightarrow\{v_1,\ldots,v_\ell\}$ such that $G\oplus_\varphi H'$ is $2$-connected. By Observation~\ref{obs:subgr}, $\varphi$ can be extended to $V(H)$ in such a way that $F=G\oplus_\varphi H$ is $2$-connected.

Assume now that $\ell$ is odd. Let $H'$ be a component of $H$ and denote by $x$ and $y$ its vertices. Consider a shortest $(v_1,v_2)$-path in $G$. Notice that this path contains a vertex $u$ that does not belong to the biconnected components $P_1$ and $P_2$. Moreover, this vertex does not belong to any pendant biconnected component of $G$. Define $\varphi'(x)=v_\ell$ and $\varphi'(y)=u$. Let $G'=G\oplus_{\varphi'}H'$. Observe that $G'$ has $\ell'=\ell-1$ pendant biconnected components $P_1',\ldots,P_{\ell'}'$. Since $\ell'$ is even, we can use the already proved claim and obtain that there is an injection $\varphi''\colon V(H'')\rightarrow V(G')$ such that  $F=G'\oplus_\varphi H''$ is $2$-connected for $H''=H-\{x,y\}$. Let \[\varphi(x)=
\begin{cases}
\varphi'(x),&\mbox{if~} x\in V(H'),\\
\varphi''(x),& \mbox{if~} x\in V(H'');
\end{cases}
\]
for $x\in V(H)$. We have that $F=G'\oplus_{\varphi''} H''=(G\oplus_{\varphi'}H')\oplus_{\varphi''}H''=G\oplus_\varphi H$ and is 2-connected.
\end{proof}

\begin{lemma}\label{lem:main}
Let $G$  and $H$ be graphs with $2\leq p(G)\leq |V(H)|\leq |V(G)|$ such that $G$ is connected,  $H$ has no isolated vertex and has a component with at least 3 vertices. Then there is an injection $\varphi\colon V(H)\rightarrow V(G)$ such that  $F=G\oplus_\varphi H$ is $2$-connected.
\end{lemma}

\begin{proof}
Let $\mathcal{P}=\{P_1,\ldots,P_\ell\}$ be the set of pendant biconnected components of $G$.  
Observe that $G$ is distinct from $K_2$. Otherwise, we have that $H=K_2$, i.e., this is a matching graph contradicting the condition that  $H$ has a component with at least 3 vertices.
We select $v_i\in V(P_i)$ for $i\in\{1,\ldots,\ell\}$ in such a way that $v_1,\ldots,v_\ell$ are pairwise nonadjacent. Notice that it always can be done because $G$ is distinct from $K_2$. 

Let $H_1,\ldots,H_r$ be the components of $H$ and assume that $|V(H_1)|\leq\ldots\leq |V(H_r)|$. Since $\ell\leq |V(H)|$, there is minimum $s\in\{1,\ldots,r\}$ such that $p=\sum_{i=1}^s|V(H_i)|\geq \ell$. Denote by $q=\sum_{i=1}^{s-1}|V(H_i)|$.
We construct the induced subgraph $H'$ of $H$ as follows.
If  $p=\ell$, then $H'$ is the subgraph of $H$ composed by the components $H_1,\ldots,H_s$. Suppose that $p>\ell$. If $\ell-q\geq 2$, then we find a connected induced subgraph $H_s'$ of $H_s$ and define $H'$ as the subgraph of $H$ with the components $H_1,\ldots,H_{s-1},H_s'$. Let $\ell-q=1$. Notice that this implies that $s\geq 2$. We consider two cases depending on $|V(H_{s-1})|$.

Suppose that  $|V(H_{s-1})|=2$. Recall that $H$ has a component $H_i$ with at least 3 vertices and $i\geq i-1$ by the ordering of the components. We find a connected induced subgraph $H_i'$ of $H_i$ with 3 vertices. Then we define $H'$ as the subgraph of $H$ with the components $H_1,\ldots,H_{s-1},H_i'$.   
 
Assume now that $t=|V(H_{s-1})|\geq 3$. We find a connected induced subgraph $H_{s-1}'$ of $H_s$ with $t-1$ vertices and a connected iduced subgraph $H_s'$ of $H_s$ with 2 vertices. Then $H'$ is the subgraph of $H$ with the components $H_1,\ldots,H_{s-2},H_{s-1}',H_s'$. 

In all the cases, $H'$ has exactly $\ell$ vertices and each component of $H'$ has at least 2 vertices. By Lemma~\ref{lem:comp-connect},  there is a bijection $\varphi\colon V(H)\rightarrow\{v_1,\ldots,v_\ell\}$ such that $F=G\oplus_\varphi H'$ is $2$-connected. By Observation~\ref{obs:subgr}, $\varphi$ can be extended to $V(H)$ in such a way that $F=G\oplus_\varphi H$ is $2$-connected.
\end{proof}

Recall that the bridges and biconnected components of a graph $G$ can be found in linear time by the algorithm of Tarjan~\cite{Tarjan74}. Combining this fact with Observations
~\ref{obs:pend-cover} and \ref{obs:two-conn} and Lemmas~\ref{lem:match} and  \ref{lem:main} we obtain the following theorem.

\begin{theorem}\label{thm:SCA-two}
Let $G$ and $H$ be graphs such that $G$ is connected, $H$ has no isolated vertices and $|V(H)|\leq|V(G)|$. 
Then there is an injective mapping $\varphi\colon V(H)\rightarrow V(G)$ such that $F=G\oplus_\varphi H$ is 2-connected if and only if  one of the following holds:
\begin{itemize}
\item[(i)] $G$ is $2$-connected,
\item[(ii)] $G$ is not $2$-connected and $p(G)\leq |V(H)|$,
\end{itemize}
unless $G$ is a star  $K_{1,n}$ where $n$ is odd and $H$ is a matching graph. 
\end{theorem}

Theorem~\ref{thm:SCA-two} immediately implies the next corollary.

\begin{corollary}\label{cor:SCA-two}
Unweighted \probSCAtwo is solvable in time $\Oh(|V(G)|+|E(G)|+|E(H)|)$.
\end{corollary}

\section{Conclusion}\label{sec:conclusion}
We initiated the investigation of the structured connectivity augmentation  problems where the aim is to increase the edge connectivity of the input graphs by adding edges when the added edges compose a given graph. In particular, we proved that \probStrucAugm and \probStrucAugmtwo are solvable in polynomial time when $H $ is from a graph class $\mathcal{C}$ with 
bounded vertex-cover number. It is natural to ask   about increasing connectivity of a $(k-1)$-connected graph to a $k$-connected graph for every positive integer $k$. For the ``traditional'' edge connectivity augmentation problem (see~\cite{Frank11,Nagamochi08}), the augmentation algorithms are based on the classic work of Dinits, Karzanov, and Lomonosov~\cite{DinicKL76} about the structure of minimum edge separators. However, for the structural augmentation, the structure of the graph $H$ is an obstacle for implementing this approach directly.  Due to this, we could not push further our approach to establish the complexity of  \probWSCA for $k>2$ when $H$ is of bounded vertex cover.  This remains a natural open question.   Recall  that our hardness results showing that it is \classNP-hard to increase  the connectivity of a $(k-1)$-connected graph to a $k$-connected graph when $H$ belongs to a class with unbounded vertex cover number are proved for every $k$. 

As the first step, it could be interesting to consider the variant of the problem for multigraphs. In this case, we allow parallel edges and assume that for a mapping $\phi\colon V(H)\rightarrow V(G)$, the multiplicity of $\phi(x)\phi(y)$ in $G\oplus_\phi H$ is the sum of the multiplicities of $\phi(x)\phi(y)$ in $G$ and $xy$ in $H$. Notice that all our algorithmic and hardness results can be restated for this variant of the problem. Actually, some of the proofs for this variant of the problem become even simpler.  
 
The question of  obtaining a $k$-connected graph for $k\geq 3$ is also open for the unweighted problem. Here we ask whether it is possible to derive   structural necessary and sufficient conditions for a $(k-1)$-connected graph $G$ and a graph $H$ such that there exists  an injective  mapping $\phi\colon V(H)\rightarrow V(G)$ such that $G\oplus_\phi H$ is $k$-connected.  

Another direction of the research is to consider vertex connectivity. As it is indicated by the existing results about vertex connectivity augmentation (see, e.g., \cite{JacksonJ05a,Jordan95}), this variant of the problem could be more complicated.


\end{document}